\documentclass[11pt,a4paper]{article}
\usepackage{eurosym}
\usepackage{amsfonts}
\usepackage{fullpage}
\usepackage[authoryear,round]{natbib}
\usepackage{amssymb,amsmath,stmaryrd,amsthm}
\usepackage{accents}
\usepackage{enumerate}
\usepackage{nicefrac}
\usepackage[mathscr]{eucal}
\usepackage{epsfig}
\usepackage{pstricks,pst-node,pst-text,pst-3d}
\usepackage{pifont}
\usepackage{textcomp}
\usepackage{multicol}
\usepackage{pst-node}
\usepackage{graphicx}
\usepackage{pstricks-add}
\usepackage{color}
\usepackage{tikz}
\usepackage{soul}
\usepackage{listings}
\usepackage[english]{babel}
\usepackage{booktabs}
\usepackage[utf8]{inputenc}
\usepackage{setspace}
\usepackage{tabularx}
\usepackage{float}
\usepackage{hyperref}
\usepackage{enumitem}

\hypersetup{hidelinks}
\theoremstyle{definition}

\newtheorem{lemma}{Lemma}
\newtheorem{remark}{Remark}

\newtheorem{example}{Example}

\newtheorem{proposition}{Proposition}

\newtheorem{corollary}{Corollary}

\def\1{\mathbf{1}}

\def\cC{\mathcal{C}}

\def\cS{\mathcal{S}}

\def \0{{\mathbf{0}}}

\def \mymid {\,:\,}
\renewcommand{\&}{\textup{\symbol{`\&}}}
\newlist{steps}{enumerate}
{1}
\setlist[steps, 1]{label = Step \arabic*:}
\makeatletter
\renewcommand{\@biblabel}[1]{#1.} 
\begin{document}

\title{Formation of international environmental agreements and payoff
allocation}
\author{Michel Grabisch${}^{(1)(2)}$\thanks{Corresponding author. Email {\tt
      michel.grabisch@univ-paris1.fr}} \and Elena Parilina${}^{(3)}$ \and Agnieszka
Rusinowska${}^{(2)}$ \and Georges Zaccour${}^{(4)}$ \\
\\
{\normalsize (1) Charles University, Department of Applied Mathematics,
Prague, Czech Republic} \\
{\normalsize (2) Universit\'{e} Paris 1 Panth\'{e}on-Sorbonne, CNRS, Centre
d'\'{E}conomie de la Sorbonne,} \\
{\normalsize Paris School of Economics, Paris, France}\\
{\normalsize (3) Saint Petersburg State University, Saint Petersburg, Russia}%
\\
{\normalsize (4) Chair in Game Theory and Management, GERAD, HEC Montreal,
Canada}}
\date{}
\maketitle

\begin{abstract}
Dealing with climate change requires that all countries engage in costly
efforts to reduce their emissions. Reaching this objective has so far been
elusive because it is in the best interest of each country to let the others
do the effort and benefit itself from a better environment. The presence of
negative externalities and strategic behavior have made game theory a
natural paradigm to design an international environmental agreement (IEA)
that codifies what countries should do. Considering that countries are
sovereign and no supranational entity can impose an agreement, a stream of
literature adopted a noncooperative mode of play to the formation of an
environmental coalition. On the other hand, as joint optimization of all
countries' payoff leads to the best outcome, cooperative games approach has
also been proposed to share the cost of climate change. Both approaches have
their pros and cons.

In this paper, we propose a model of coalition formation that combines both
cooperative and noncooperative modes of play. Starting from any given
coalition, we implement a Markov process that shows sequentially which
countries join or leave the coalition until reaching an absorbing state. All
possible sequential scenarios are considered and an allocation to the player
is made taking into account individual rationality. An illustration with
vulnerable and invulnerable countries to pollution is given.
\end{abstract}


\vspace{2 mm}


\section{Introduction}

In this paper, we propose a Markov process approach to the formation of an
international environmental agreement (IEA) that combines both cooperative
and noncooperative games aspects. IEAs are concerned by transboundary (or
global) pollution problems where the emissions generated by any country
affect the well-being of all. The global environment is a public good that
suffers from the lack of sufficient efforts by the agents (polluting
countries) to keep it clean. For each country to internalize the cost of its
emissions, it must either be incentivized through carrots, e.g., transfers
or subsidies, or through sticks, e.g., taxes and quotas. The absence of a
supranational environmental agency makes the implementation of such
instruments nearly impossible in practice in the context of global
pollution. The thinking is that an IEA could avoid, or at least alleviate,
the impact of this tragedy of the common type of situation.

From the above, it is clear that any design of an IEA must at least account
for the facts that the countries' payoffs are interdependent (by the very
public good nature of the environment), and that each country behaves
strategically by choosing its pollution level considering the others'
reactions. These features have made game theory a natural paradigm to model
IEAs, which explains the large number of publications in this area during
the last three decades or so. The literature can be classified along two
main dimensions: (i) the mode of play, i.e., cooperative or noncooperative,
and (ii) the temporal nature of the process of coalition formation, i.e.,
static or dynamic. We review below the part of the literature that is
directly relevant to our approach, and next highlight our contribution.

\subsection{Review of the literature}

\textbf{Noncooperative game approach to IEAs.} The premise in this
literature stream is straightforward: In the absence of an international
institution that can force sovereign countries to join an IEA, the
participation must be voluntary based, and the agreement self-enforcing
(stable in an equilibrium sense). The objective function of a country is
typically made of two pieces: a concave increasing benefit function that
depends only on own emissions, which are a by-product of production, and a
convex increasing damage cost function that depends on total emissions (here
lies the negative externality). The endogenous size of a stable agreement is
determined by two rules referred to as internal and external stability. An
IEA is internally stable if no member can improve its payoff by leaving it,
and externally stable if no non-member can achieve a better outcome by
joining the agreement. These rules were originally proposed in \cite%
{Aspremont et al.} to study the stability of collusive pricing by a cartel.
The payoffs at the basis of the comparisons involved in the no-entry and
no-exit conditions are computed assuming that the members (signatories) act
as a single player and choose the emissions levels that jointly maximize
their welfare, while each non-member (non-signatory) optimizes its
individual payoff.

The seminal papers in this stream are \cite{Barrett94} and \cite%
{CarraroSiniscalco93}, where the authors showed that, without additional
incentives, only a few (typically 2 or 3) countries will join an IEA. \cite%
{DiamandSartz06} showed that, independently of the number of countries, the
IEA will include at most 4 countries when the benefit and cost functions are
quadratic, an assumption often made in the literature. The reason behind
this grim result is free riding. Indeed, each country is better off letting
the others invest in costly pollution abatement and enjoys a better
environment without paying itself the cost. In these seminal papers, and in
most of the literature that followed, the countries are symmetric.
Consequently, it does not really matter which country joins and which does
not, the main point being how many of the $n$ countries will end up signing
the agreement. \cite{Pavlova} assumed asymmetries in both the benefit and
cost functions and obtain that large stable coalitions will materialize only
if they include countries that have relatively high marginal benefits and a
relatively low marginal environmental damage, which is uninteresting as
these countries do not contribute much to the common good. That is, we have
a larger IEA but not much less emissions. \cite{Ansink et al.} proposed a
four-stage model of coalition formation with a noncooperative setting and
used the concept of subgame perfect Nash equilibrium to find the IEA. The
feature of the model is that the players can differently participate in the
agreement by choosing of being a member, supporter, and a free-rider. Like
in the standard IEA game, members coordinate their abatement levels and all
other countries (supporters and free-riders) act as singletons. The authors
show that under mild conditions on the costs and benefits of contributing to
the public good (e.g., abatement of greenhouse gas emissions), there exist
equilibria with support.

In an attempt to enlarge the size of the agreement, scholars extended the
basic model in various directions, such as Stackelberg leadership (e.g., 
\cite{DiamandSartz06}, \cite{RubioUlph06}), transfer schemes (\cite%
{HoelSchneider97}, \cite{FuentesRubio10}), reputation effects (e.g., \cite%
{HoelSchneider97}, \cite{JeppesenAndersen98}, \cite{CabonRamani06}, \cite%
{SaccoZaccour19}), issue linkages (e.g., \cite{Barrett97b}, \cite%
{BotteonCarraro98}, \cite{Katsoulacos97}, \cite%
{CarraroSiniscalco97,CarraroSiniscalco98}, \cite{LeBretonSoubeyran97}, \cite%
{Mohr}), punishments (\cite{HoelSchneider97}, \cite{Bahn et al.}, \cite%
{key-5}), imposing a ratification threshold (see, e.g., \cite%
{CourtoisHaeringer05}, \cite{RubioCasino05}, \cite%
{CarraroMarchioriOreffice09}, \cite{WeikardWanglerFreytag15}), seeking a
regional rather than global agreements (\cite{Asheim et al.}), considering
modest emission-reduction targets (\cite{FinusMaus08}); see also \cite%
{Barrett03} for a general discussion.

Finally, we note that a series of contributions considered investments in
both emissions abatement and adaptation to climate change (e.g., \cite%
{Lazkano et al.}, \cite{BS2017,BS2019ERE,BS2019}, \cite{maza18}, \cite%
{RohrerRubio24}). For additional references on noncooperative IEA games, we
refer the interested reader to the surveys in \cite{Finus01,Finus08}, \cite%
{Wagner01}, \cite{Mitchell}, \cite{Jorgensen2010433}, \cite%
{BenchekrounLong2012}, \cite{Carraro2013}, \cite{MarrouchChaudhuri16}.

\medskip

\textbf{Process of coalition formation.} The IEA game is (typically) played
in two stages. In the first stage, a coalition of $s$ countries forms, while
the remaining $n-s$ countries stay out. In the second stage, a
noncooperative emissions game is played to determine the equilibrium
emissions and payoffs. Recall that the coalition acts as a single player and
each of the $n-s$ non-member players optimizes individually its payoff. One
checks if this coalition satisfies the internal and external stability
conditions. If it is the case, then we have a stable IEA. Next, set $s=s+1$
and redo the two steps.

The process just described assumes that the players are myopic. \cite%
{DiamandSartz06} endow countries with foresightedness. If a country
contemplates leaving or joining a coalition, then it will determine its
benefit considering its action could incentivize other countries to defect
from the coalition or to join it. The implication is that the resulting
coalition is not necessarily of size $s-1$ or $s+1$ but could be of any
size. The authors derive the conditions for the existence of a unique set of
farsighted stable IEAs.

In the above cited papers, the emissions game is static. A series of
contributions retained a dynamic game model to capture the fact that the
environmental damage is caused by the accumulation of pollution, modeled as
a state variable, and by current one. The membership game could be played
once or at each stage. \cite{key-5} proposed an evolutionary game model,
with a state equation governing the evolution of the pollution stock, and a
replicator dynamics equation describing the evolution of the membership.
Interestingly, the authors obtain that, depending on the parameter values, a
variety of stable coalitions can materialize, including the grand coalition.

\cite{RubioUlph07} retained a dynamic game model in discrete time where in
each stage, the players solve both an emissions game and a membership game
so that the size of the coalition varies over time with the pollution stock.
At the unique steady state, we have both the steady-state value of the
pollution stock and the stable IEA. Note that the authors assume the
existence of an upper bound on non-signatory emissions and that signatories
are randomly selected in each stage. \cite{RubioUlph07} extends \cite%
{RubioCasino05} 
where the membership game is played once. \cite{Zeeuw08} extends the static
farsightedness behavior process in \cite{DiamandSartz06} to a dynamic game
including a pollution stock. In \cite{RubioUlph07} and \cite{Zeeuw08}, both
the membership game and the emissions game are dynamic.

A series of papers imbed the two-stage participation game into a repeated
game, see, e.g., \cite{Barrett03}, \cite{Bagga}. In the later, see also \cite%
{Schmidt2017}, signatories can commit to the number of periods during which
the agreement is binding. \cite{Karp} consider a dynamic model where
countries cannot make long-term commitments, but can communicate with each
others to form beliefs about the consequences of re-opening negotiations.
Interestingly, the authors obtain that an effective agreement can emerge.

\medskip

\textbf{Cooperative game approach to IEAs.} \cite{PloegZeeuw92} and \cite%
{DocknerLong93} were among the firsts to characterize and contrast
cooperative and noncooperative solutions to a dynamic (differential) game
model of pollution control. They showed that cooperation leads to a higher
collective welfare than the sum of individual welfares, and total
cooperative emissions are lower than their noncooperative counterparts.
These results are an invitation to design an IEA that includes all
countries. This is the starting point of the cooperative game literature on
IEAs. \cite{chtu92,chtu95,chtu97} introduced the concept of $\gamma$-core,
which was later applied to IEAs in, e.g., \cite{Germain et al.}. In the $%
\gamma$-characteristic function concept, the payoff of a coalition is given
by the Nash equilibrium outcome in the noncooperative game played by the
coalition (acting as a single player) and the non-members acting
individually. Interesting to note that this is also how the emissions game
is played in the noncooperative approach. \cite{EyckmansTulkens03} provide
some simulations for assessing the coalitionally stable sharing of the cost
of climate change. \cite{BretonFredjZaccour06} discuss different definitions
of characteristic functions and their impact on deterring free riding.

A cooperative outcome is collectively optimal, but is not an equilibrium,
and therefore not self-supported. A country can deviate from the grand
coalition and achieve a higher payoff than the one it obtains by remaining
in the coalition. To make a cooperative solution an equilibrium, one can use
trigger strategies that punish harshly and credibly any deviation from the
cooperative trajectory. This established result in repeated and state-space
games of infinite duration might not be easy to apply in the international
arena. Chander and Tulkens did not rely on this argument to sustain their $%
\gamma $-core IEA, but on the fact that any deviation would lead to the
grand coalition to fall apart, which makes the deviation uninteresting,
i.e., leads to a lower payoff.\footnote{%
In the $\gamma$ characteristic function concept, the value of a coalition is
its Nash equilibrium outcome in the noncooperative game between the
coalition, acting as one player, and the nonmembers acting individually.}
Whereas here we have either an IEA including all players or no coalition at
all, partial coalitions are feasible in the noncooperative approach.

\cite{PetrosjanZaccour03} propose a time-consistent Shapley value to share
the cost of an IEA among the $n$ countries. As in the $\gamma$-core
approach, the sustainability of the grand coalition is supported by the fact
that a deviation along the cooperative state trajectory at any date before
reaching the planning horizon will lead all players to switch to their
noncooperative emissions in the subgame starting at that date. As the
Shapley value is an imputation, it is not rational for any player to deviate.

To wrap up, the (typically) small coalition outcome of a noncooperative game
approach to the formation of an IEA is entirely driven by the internal and
external stability concepts. The IEA inclusive of all countries obtained in
a cooperative game approach is driven by the view that any deviation will
destroy completely the agreement. \cite{DiamandSartz06} argue that
farsightedness constitutes a vision between the two. Indeed, when a country
considers deviating, it should not assume optimistically (myopically) that
the coalition will remain intact, nor pessimistically (they call it
phobically) that the coalition will completely collapse.

\subsection{Contribution of our paper}

A first originality of our approach is that it combines both cooperative and
noncooperative games thinking to design an IEA, and not either of the two
modes of play as it is done in the literature. Indeed, we use both the
internal and external stability conditions and we define a sharing mechanism
of coalition collective outcome among its members.

\begin{remark}
In two-stage IEA formation games, the members cooperate in setting their
emissions strategies, but the game is still played noncooperatively, that
is, the final outcome is an equilibrium and not an allocation in a
cooperative game meaning. The same applies to the three-stage game in \cite%
{Barrett 2005} where the signatories choose collectively their abatement
targets in the second stage and non-signatories individually determine their
abatements in the third stage. Although coordination gives a cooperative
game flavor to the process, there is no use of a cooperative game solution
that is applied either to the endogenously determined stable coalition, or
to the grand coalition. Here, we do use a cooperative game solution.
\end{remark}

The second original feature is that our coalition formation process is
sequential and visible, that is, at each step, we determine endogenously
which country has an interest in joining or leaving the coalition formed at
the preceding step. As we will see, this defines a Markov process of
coalition formation, and any absorbing state of the Markov chain is a stable
IEA. We show that absorbing classes must be absorbing states or periodic
classes.

The third original feature concerns the sharing of benefits among
participants in the IEA, once an absorbing state is reached. All classical
cooperative game solutions, e.g., the Shapley value, the core, and the
nucleolus, assume that the grand coalition will form, which is far from
being the case in the context of an IEA. Instead of using such solutions, we
propose here to implement the Markovian process value introduced in \cite%
{fagr12,fagr13}, which perfectly fits our situation and, interestingly, is a
generalization of the Shapley value. Indeed, supposing that the process of
formation of a coalition obeys a Markov chain, \cite{fagr12} proposed the
notion of scenario-value, which gives the marginal contribution of every
player along the scenario. The process-value is the expected value of the
scenario-values, taken over all possible scenarios\footnote{%
A scenario is any sequence of coalitions generated by the Markov chain
starting from some fixed coalition, e.g., the empty coalition, and reaching
an absorbing state.} with the probability of their occurrence. 
Therefore, it does not matter whether the final coalition (the IEA) is the
grand coalition or not, and all possibilities of reaching the final
coalition are taken into account in the computation. We derive some general
properties of the Markovian process value applied to the formation of IEAs
and show its flexibility through a detailed example. We note that the
Markovian process value is axiomatized, which provides a normative
foundation to the method of sharing benefits we are proposing.

Further, our approach integrates the idea that the players are farsighted,
which we implement by considering only those scenarios leading to a
scenario-value that is an imputation, i.e., any country knows that entering
this scenario will lead to a payoff at least equal to what it can get by
acting alone.

\begin{remark}
In our framework, farsightedness is a natural request as it is synonymous of
individual rationality, which is a corner stone aspect in both cooperative
and noncooperative game theory.
\end{remark}

Finally, in the example of formation of an IEA, 
we consider the following research question: What is the impact on coalition
formation if a country must pay a cost when deciding to leave an agreement?
This cost can be justified by the harm caused to the other signatories,
e.g., they need to support the operating costs of the agreement with less
members.\footnote{%
We recall that Canada pulled out of the Kyoto Protocol after signing it and
the US exited the Paris Agreement, so leaving an agreement after signing it
has an empirical support.} 
Here, one contribution of our results is that introducing such a cost can
lead to larger stable coalition. An alternative interpretation of this cost
would be the price to pay to a country in the IEA because staying in
incentivizes others to join.

The rest of the paper is organized as follows. In Section \ref{sec:mark}, we
recall the Markovian process value, and introduce in Section \ref{sec:fram}
the environmental game model. We discuss the IEA formation process in
Section \ref{sec:IEA formation} and provide a detailed example in Section %
\ref{sec:example}. We briefly conclude in Section \ref{sec:conclusions}.

\section{The Markovian process value\protect\footnote{%
This section draws heavily on \cite{fagr12,fagr13}.}}

\label{sec:mark}

In this section, we introduce a value for TU-games suitable for Markovian
coalition formation processes, which encompasses the classical Shapley
value. We refer the reader to \cite{fagr12,fagr13} for a full exposition.

We consider a random process of coalition formation among players from set $%
N $, supposing that the process is Markovian and homogeneous, i.e., its
transition matrix does not depend on time. Specifically, the set of states
is the set of all possible coalitions $2^N$, and for any $S,T\in 2^N$, the
transition from $S$ to $T$, denoted by $S\rightarrow T$, occurs with
probability $p_{S,T}$ and does not depend on the history. We denote by $%
P=[p_{S,T}]$ the transition matrix, which is row-stochastic.

Fixing a finite horizon $\tau$, a \textit{scenario} $\mathcal{S}$ of length $%
t\leqslant \tau$ is a sequence of coalitions $S_0S_1\cdots S_t$ with $S_i\in
2^N$, $i=0,\ldots,t$, which can be produced by the Markov chain, i.e., $%
p_{S_i,S_{i+1}}>0$ for every $i=0,\ldots,t-1$. The probability of its
occurrence is simply 
\begin{equation}  \label{prob occurrence}
\mathbb{P}(\mathcal{S}) = \mathbb{P}(S_0)\prod_{i=0}^{t-1}p_{S_i,S_{i+1}},
\end{equation}
where $\mathbb{P}(S_0)$ is the probability of starting from coalition $S_0$.
Note that there may be repetitions of coalitions in the scenario, $S_t$ is
not necessarily the grand coalition $N$, and $S_0$ is not necessarily the
empty set. Also, we take the convention that if $S$ is a terminal state of
the Markov chain (i.e., $p_{S,S}=1$), then any scenario containing $S$
terminates at $S$, deleting the repetitions $SS\cdots$.

Denote by $S\Delta T:=(S\cup T)\setminus (S\cap T)$ the symmetric difference
of $S$ and $T$. We call \textit{simple scenario} any scenario $\mathcal{S}%
=S_{0}S_{1}\cdots S_{t}$ such that $|S_{i}\Delta S_{i+1}|\leqslant 1$ for $%
i=0,\ldots ,t$, i.e., in a transition at most one player is \textit{active},
which means either entering the coalition $S_{i+1}$ or leaving the coalition 
$S_{i}$. Accordingly, we call \textit{simple} a coalition formation process
producing only simple scenarios.

Let $\mathcal{G}(N)$ be the set of TU-games on $N$. Given a simple Markov
coalition process with transition matrix $P$, denote by $\mathfrak{S}(P,\tau
,S_{0})$ the set of all (simple) scenarios produced by the Markov process of
length at most $\tau $ and starting at $S_{0}$. For any simple scenario $%
\mathcal{S}=S_{0}S_{1}\cdots S_{t}\in \mathfrak{S}(P,\tau ,S_{0})$, we
define the \textit{scenario-value} $\phi ^{\mathcal{S}}:\mathcal{G}%
(N)\rightarrow {\mathbb{R}}^{N}$ by 
\begin{equation}
\phi _{i}^{\mathcal{S}}(v)=\sum_{k\mymid S_{k}\Delta S_{k+1}\ni
i}(v(S_{k+1})-v(S_{k})),\quad i\in N.  \label{eq:scenval}
\end{equation}%
(For the general case with nonsimple scenarios, see \cite{fagr12,fagr13}.)
Then, the \textit{process-value} $\phi :\mathcal{G}(N)\rightarrow {\mathbb{R}%
}^{N}$ for $\tau ,S_{0}$ is defined by 
\begin{equation}
\phi _{i}(v)=\sum_{\mathcal{S}\in \mathfrak{S}(P,\tau ,S_{0})}\mathbb{P}(%
\mathcal{S})\phi _{i}^{\mathcal{S}}(v),\quad i\in N.  \label{eq:procval}
\end{equation}%
Note that since all scenarios are beginning with $S_{0}$, we assume $\mathbb{%
P}(S_{0})=1$. Also, $S_{0}$ must be chosen such that $p_{S_{0},T}>0$ at
least for one coalition $T\neq \varnothing $ such that $|S_{0}\Delta T|=1$.
A natural choice is to choose $S_{0}$ such that $p_{T,S_{0}}=0$ for all $%
T\neq \varnothing $, i.e., $S_{0}$ is a \textquotedblleft
source\textquotedblright\ in graph-theoretic terms. One may also take
several starting points $S_{0}$ with some probability and consider the
expected value.

We show now that the classical Shapley value is a particular case of the
process-value. Consider the Markov process defined by the transition matrix 
\begin{equation*}
p_{S,T}=%
\begin{cases}
\displaystyle\frac{1}{n-|S|}, & \text{if }S\subseteq T,|T\setminus S|=1, \\ 
1, & \text{if }S=T=N, \\ 
0, & \text{otherwise}.%
\end{cases}%
\end{equation*}%
It is easy to see that there are $n!$ scenarios of length $n$ produced by
this process starting from $\varnothing $ and ending at $N$, with equal
probability, which corresponds to the maximal chains in the Boolean lattice $%
(2^{N},\subseteq)$. Therefore, the process-value coincides with the Shapley
value \citep{fagr12}.

\section{The environmental game model}

\label{sec:fram}

Consider a set $N=\{1,\ldots ,n\}$ of countries (players) contemplating the
formation of an international environmental agreement (IEA). Each country $i$
produces a quantity of goods. Assuming a monotone relationship between
production and the level of carbon emission $e_{i}$, we can express the
revenue of country $i$ as a function of $e_{i}$. Let $e=(e_{1},\ldots
,e_{n}) $. Following 
\cite{PloegZeeuw92}, we assume that the payoff function of country $i$ has
the form 
\begin{equation}
\pi _{i}(e)=B_{i}(e_{i})-D_{i}\Big( \sum_{j\in N}e_{j}\Big),  \label{eq:1}
\end{equation}%
where $B_{i}$ is the benefit function from economic activity induced by $%
e_{i}$, and $D_{i}$ is the damage function. We assume that $B_{i}(e_{i})$ is
concave increasing and $D_{i}\left( \cdot \right) $ is convex increasing.
Each country $i$ aims at maximizing its payoff $\pi _{i}$ by choosing the
level of emissions $e_{i}$. As the damage function depends on the total
emissions, each player's payoff depends on other players' decisions.

Suppose that the formation of an IEA can be modeled as a cooperative game.
Consequently, we need first to determine the optimal collective outcome that
is achieved by maximizing the sum of all countries' payoffs, and second to
allocate this outcome to the players by using some solution concept of
cooperative games. Instead of using one of the classical solution concepts,
e.g., the Shapley value, the core, or the nucleolus, we propose to use a
solution concept taking into account the process itself of coalition
formation, which is the Markovian process value, described in Section~\ref%
{sec:mark}.

For this second step, we must first compute the characteristic function (CF)
values of all possible coalitions. Recall that the CF $v\left( \cdot \right) 
$ is defined by 
\begin{equation*}
v: 2^N \rightarrow \mathbb{R},\text{ \ \ \ \ }v\left( \varnothing \right) =0.
\end{equation*}
We adopt here the $\gamma $-CF due to \cite{chtu97}. In this setup, the $%
\gamma$-CF, denoted $v^{\gamma}\left( S\right)$, where $S\subseteq N$ is a
coalition, is defined as the partial (Nash) equilibrium outcome of the
noncooperative game between the members of coalition $S$, acting as a single
player, and non-members acting individually. Formally, for any nonempty
coalition $S\subseteq N$, we have 
\begin{align}
v^{\gamma }(S)& :=\sum_{i\in S}\pi _{i}(e_{S}^{\gamma },\{e_{j}^{\gamma
}\}_{j\in N\backslash S})  \label{eq:gamma} \\
e_{S}^{\gamma }& :=\arg \max_{e_{S}}\sum_{i\in S}\pi
_{i}(e_{S},\{e_{j}^{\gamma }\}_{j\in N\backslash S})  \notag \\
e_{j}^{\gamma }& :=\arg \max_{e_{j}}\pi _{j}(e_{S}^{\gamma
},e_{j},\{e_{i}^{\gamma }\}_{i\in N\backslash \{S\cup j\}}),~\text{for all }%
j\in N\backslash S,  \notag
\end{align}%
where $e_{S}=\{e_{i}\}_{i\in S}$, that is, the vector of emissions of the
players belonging to coalition $S$. For $S=\varnothing $, as this means that
the coalition formation process has not yet started, we set by convention $%
\pi _{i}(e_{\varnothing }^{\gamma },\{e_{j}^{\gamma }\}_{j\in N})=0$ for all 
$i\in N$. This implies $v(\varnothing )=0$, as expected.

From now on, we drop the superscript $\gamma $ from $v$, and for simplicity
write $e^{S}:=(e_{S}^{\gamma },\{e_{j}^{\gamma }\}_{j\in N\backslash S})$.

\section{The IEA formation process}

\label{sec:IEA formation}

\subsection{The transition mechanism}

\label{subsec:transition}

Consider that a coalition $S$ of countries has formed. Any country $i$ then
faces either of the two questions: If $i\in S$, then is it payoff-improving
to leave it? If $i\not\in S$, then is it in its best interest to join $S$?
To answer these questions, we consider the following quantities: 
\begin{align}
\delta_{i}^{-}(S)& =\pi_{i}(e^{S})-\pi_{i}(e^{S\setminus \{i\}}) +c,\quad
i\in S, S\subseteq N,|S|>1  \label{eq:d-} \\
\delta_{i}^{-}(\{i\}) &=\pi_{i}(e^{\{i\}})  \label{eq:0} \\
\delta_{i}^{+}(S)& =\pi_{i}(e^{S\cup \{i\}})-\pi_{i}(e^{S}),\quad i\not\in
S, S\subseteq N  \label{eq:d+}
\end{align}
where $c\geqslant 0$ represents a cost of leaving the current  coalition.
Observe that due to our convention, $\delta _{i}^{+}(\varnothing )=\pi
_{i}(e^{\{i\}})$, $i\in N$.

Player $i\in N\setminus S$ is better off joining$\frac{{}}{{}}$ $S$  if $%
\delta _{i}^{+}(S)>0$: indeed, the payoff of $i$ is greater if $i$ joins  $S$
than if staying outside $S$. On the other hand, player $i\in S$ has an
incentive to leave  $S$ if $\delta _{i}^{-}(S)<0$, which means that the gain
of payoff if leaving  the coalition is greater that the cost incurred by
leaving. This cost can be  interpreted either as a money amount or as a
damage to the reputation of  country $i$, due to leaving the current
international environmental agreement  for selfish reasons. Note that in the
literature, usually no such cost is  considered ($c=0$). We highlight that
there is no cost when a player leaves the coalition of size 1 (see equation %
\eqref{eq:0}), where he is the unique player. The reason we assume this is
that there is no reputation damage or other expenses when such a coalition
is destructed.

We call \textit{state} the current coalition. The quantities in (\ref{eq:d-}%
)-(\ref{eq:d+}) govern the process of coalition formation, starting from the
empty coalition, as follows. For every state $S$:

\begin{itemize}
\item If for some $i\in S$, $\delta_i^-(S)<0$, then $i$ leaves $S$ with
probability proportional to $|\delta_i^-(S)|$: 
\begin{equation}  \label{eq:p1}
p_{S,S\setminus\{i\}} = \frac{|\delta_i^-(S)|}{\sum_{\substack{ j\in
N\setminus S  \\ \delta^+_j(S)>0}}\delta_j^+(S) + \sum_{\substack{ j\in S 
\\ \delta^-_j(S)<0}}|\delta_j^-(S)|}.
\end{equation}

\item If for some $i\not\in S$, $\delta_i^+(S)>0$, then $i$ enters $S$ with
probability proportional to $\delta_i^+(S)$: 
\begin{equation}  \label{eq:p2}
p_{S,S\cup \{i\}} = \frac{\delta_i^+(S)}{\sum_{\substack{ j\in N\setminus S 
\\ \delta^+_j(S)>0}}\delta_j^+(S) + \sum_{\substack{ j\in S  \\ %
\delta^-_j(S)<0}}|\delta_j^-(S)|}.
\end{equation}

\item Otherwise, i.e., if for all $i\in S$, $\delta _{i}^{-}(S)\geqslant 0$,
and for all $i\not\in S$, $\delta _{i}^{+}(S)\leqslant 0$, the process
remains in state $S$: 
\begin{equation}
p_{S,S}=1.  \label{eq:p3}
\end{equation}
\end{itemize}

Equations (\ref{eq:p1})-(\ref{eq:p3}) define a Markov chain on the set of
states $2^{N}$, with transition probabilities $p_{S,T}$ from a state $S$ to
a state $T$, where $T$ is either $S\cup \{i\}$ or $S\setminus \{i\}$ for
some country $i$. 

Knowing the transition matrix $P=[p_{S,T}]$ permits to deduce all possible
scenarios of coalition formation that can happen, and to compute the
probability of their occurrence. Observe that by construction all scenarios
are simple, and therefore, the scenario-values and the process-value can be
computed according to (\ref{eq:scenval}) and (\ref{eq:procval}).

We list hereunder some properties of the transition graph $G$ obtained from
the transition matrix $P$.

\begin{lemma}
\label{lem:tg} The following holds:

\begin{enumerate}
\item Except absorbing states, no state can have a loop, i.e., a transition
to itself.

\item Given two states $S$ and $T$, there is a transition from $S$ to $T$ or
from $T$ to $S$ only if $|S\Delta T|=1$. In addition, if there is a
transition from $S$ to $T$, there cannot be a transition from $T$ to $S$. In
other words, there is no cycle of length 2 in $G$.

\item There is no cycle of odd length, except loops (length = 1).

\item Cycles of length 4 which are absorbing classes exist, for any number
of countries (see example below with 3 countries).

\psset{unit=0.8cm} \pspicture(-0.2,-1)(6.8,5) \cnodeput(0,2){0}{$\varnothing 
$} \cnodeput(2,4){1}{$1$} \cnodeput(2,2){2}{$2$} \cnodeput(2,0){3}{$3$} %
\cnodeput(4,4){12}{$12$} \cnodeput(4,2){13}{$13$} \cnodeput(4,0){23}{$23$} %
\cnodeput(6,2){123}{$123$} \ncline[linecolor=red]{->}{0}{1} %
\ncline[linecolor=red]{->}{2}{0} \ncline{->}{3}{0} %
\ncline[linecolor=red]{->}{1}{12} \ncline[linecolor=red]{->}{12}{2} %
\ncline{->}{13}{1} \ncline{->}{23}{2} \ncline{->}{13}{3} \ncline{->}{23}{3} %
\ncline{->}{123}{12} \ncline{->}{13}{123} \ncline{->}{23}{123} 
\endpspicture
\end{enumerate}
\end{lemma}

\begin{proof}
  \begin{enumerate}
  \item  This is by construction of the transition matrix (formulas
    (\ref{eq:p1}) to (\ref{eq:p3})).
  \item  The first assertion comes from the fact that the process
    considers that only one player can enter or leave at a time.
    The second assertion is true because $S\Delta T$ is a single
  country $i$, and a transition from $S$ to $T$ means that for country $i$, its
  payoff is better under the state $S$ than in the state~$T$.
\item Suppose a cycle of odd length exists, say $S_1\rightarrow
  S_2\rightarrow\cdots\rightarrow S_{2k}=S_1$, with $k>0$. Denote by $s$ the
  cardinality of $S_1$ and suppose it is an odd number. As all transitions
  involve exactly one country, the cardinality of $S_2$ is even, and therefore
  the cardinality of $S_3$ is odd again. It follows that eventually the
  cardinality of $S_{2k}$ is even, which is impossible. The other case works the
  same. 
\item Consider the cycle $\varnothing\rightarrow 1\rightarrow 12 \rightarrow 2
  \rightarrow \varnothing$. It obeys the properties given in item 1, and hence is a
  possible transition graph. Now, to make it an absorbing class, it suffices
  that for 1, all transitions with $1i$, $i\neq 2$, are in the direction
  $1i\rightarrow 1$, for 12, all transitions with $12i$, $i\neq 1,2$, are in the
  direction $12i\rightarrow 12$, and similarly for 2 and $\varnothing$.
  \end{enumerate}
\end{proof}


We recall that an absorbing class is \textit{aperiodic} if the greatest
common divisor of the length of all cycles in the class is 1; otherwise, the
class is said to be \textit{periodic}, with the period being the greatest
common divisor. The following corollary is a direct consequence of the above
Lemma, and describes which types of absorbing class we can have.

\begin{corollary}
\label{cor:abcl} Consider the Markov process induced by some decisive payoff
functions $\pi _{1},\ldots ,\pi _{n}$ as described. Then, the absorbing
classes are either singletons (absorbing states) or periodic classes.
\end{corollary}

\begin{proof}
Consider an absorbing class $\cC$. If $|\cC|=1$, then it reduces to an absorbing
state (examples below will show that this situation is possible). Suppose now
that $|\cC|>1$. As shown in Lemma~\ref{lem:tg} (Item 3), there is no cycle of odd
length, except loops. Therefore, the greatest common divisor of the length of
all cycles cannot be 1 and equals at least 2. Consequently, all classes not reduced
to a single state are periodic.
\end{proof}

\subsection{Myopic and farsighted behavior}

A country $i$ is \textit{myopic} if it decides to enter or leave a coalition
only considering (\ref{eq:d-}) and (\ref{eq:d+}). If all countries are
myopic, an assumption that has been largely adopted in the literature, the
process of coalition formation will exactly follow the Markov process
described above.

A country $i$ is said to be \textit{farsighted till horizon }$\tau $ if it
engages in a scenario $\mathcal{S}\in \mathfrak{S}(P,\tau ,\varnothing )$
only if after $\tau $ steps, the scenario-value for country $i$ satisfies $%
\phi _{i}^{\mathcal{S}}(v)\geqslant v(\{i\})$. Such a scenario is called 
\textit{feasible {for country }}$i$. {A scenario is \textit{feasible} if it
is feasible for every country. Observe that to be farsighted at horizon $%
\tau =1$ is equivalent to be myopic, because in our context the first
transition is always of the form $\varnothing \rightarrow \{i\}$. Then,
limiting the scenario to the first step yields $\phi _{i}^{\mathcal{S}%
}(v)=v(\{i\})\geqslant v(\{i\})$, which is always satisfied.}

\begin{remark}
Feasibility is similar to individual rationality in cooperative game theory
with a nuance. In a classical cooperative game, it is the sharing of the
grand coalition's payoff that must be individually rational, i.e., each
player gets at least $v(\{i\})$. Here, it is the sharing of the (not
necessarily grand) coalition's payoff that must be individually rational.
\end{remark}

\begin{example}
Let us give an illustrative example, considering $n=3$, and the transition
graph given below.

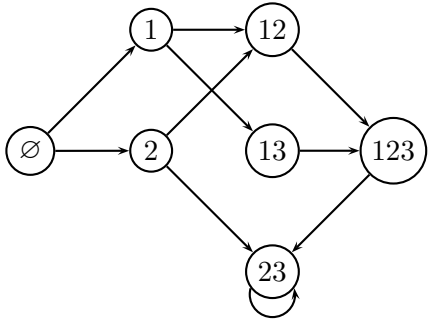
\psset{unit=0.8cm} \pspicture(-0.2,-1)(6.8,5) \cnodeput(0,2){0}{$\varnothing$%
} \cnodeput(2,4){1}{$1$} \cnodeput(2,2){2}{$2$} \cnodeput(4,4){12}{$12$} %
\cnodeput(4,2){13}{$13$} \cnodeput(4,0){23}{$23$} \cnodeput(6,2){123}{$123$} %
\ncline{->}{0}{1} \ncline{->}{0}{2} \ncline{->}{1}{12} \ncline{->}{1}{13} %
\ncline{->}{12}{123} \ncline{->}{13}{123} \ncline{->}{2}{12} %
\ncline{->}{2}{23} \ncline{->}{123}{23} \nccircle{->}{23}{-0.3cm} 
\endpspicture

There are four scenarios, which are: 
\begin{align*}
\mathcal{S}_1: &\quad \varnothing\rightarrow 1\rightarrow 12\rightarrow
123\rightarrow 23 \\
\mathcal{S}_2: &\quad \varnothing\rightarrow 1\rightarrow 13\rightarrow
123\rightarrow 23 \\
\mathcal{S}_3: &\quad \varnothing\rightarrow 2\rightarrow 12\rightarrow
123\rightarrow 23 \\
\mathcal{S}_4: &\quad \varnothing\rightarrow 2\rightarrow 23
\end{align*}
Assuming the horizon is $\tau=4$ and using formula \eqref{eq:scenval} to
calculate the scenario-values, the conditions for feasibility of the
scenarios are:

\begin{itemize}
\item Scenario $\mathcal{S}_1$: 
\begin{align*}
& \phi^{\mathcal{S}_1}_1 = v(1)+v(23)-v(123)\geqslant v(1)\Leftrightarrow
v(23)\geqslant v(123), \\
& \phi^{\mathcal{S}_1}_2 = v(12)-v(1)\geqslant v(2)\Leftrightarrow
v(12)\geqslant v(1)+v(2), \\
& \phi^{\mathcal{S}_1}_3 = v(123)-v(12)\geqslant v(3)\Leftrightarrow
v(123)\geqslant v(12)+v(3).
\end{align*}

\item Scenario $\mathcal{S}_2$: 
\begin{align*}
& \phi^{\mathcal{S}_2}_1 = v(1)+v(23)-v(123)\geqslant v(1)\Leftrightarrow
v(23)\geqslant v(123), \\
& \phi^{\mathcal{S}_2}_2 = v(123)-v(13)\geqslant v(2)\Leftrightarrow
v(123)\geqslant v(13)+v(2), \\
& \phi^{\mathcal{S}_2}_3 = v(13)-v(1)\geqslant v(3)\Leftrightarrow
v(13)\geqslant v(1)+v(3).
\end{align*}

\item Scenario $\mathcal{S}_3$: 
\begin{align*}
& \phi^{\mathcal{S}_3}_1 = v(12)-v(2)+v(23)-v(123)\geqslant v(1), \\
& \phi^{\mathcal{S}_3}_2 = v(2)\geqslant v(2) \text{ (always true)}, \\
& \phi^{\mathcal{S}_3}_3 = v(123)-v(12)\geqslant v(3)\Leftrightarrow
v(123)\geqslant v(12)+v(3).
\end{align*}

\item Scenario $\mathcal{S}_4$: 
\begin{align*}
& \phi^{\mathcal{S}_4}_1 =0 \text{ (country 1 is inactive)}, \\
& \phi^{\mathcal{S}_4}_2 = v(2)\geqslant v(2) \text{ (always true)}, \\
& \phi^{\mathcal{S}_4}_3 = v(23)-v(2)\geqslant v(3)\Leftrightarrow
v(23)\geqslant v(2)+v(3).
\end{align*}
\end{itemize}

Observe that if $v$ is superadditive, all scenarios are feasible for
countries 2 and 3. Feasibility for country 1 involves more constraints,
possibly compatible with superadditivity, unless $v$ is nonnegative. Indeed, 
$v(123)\geqslant v(1)+v(23)\geqslant v(23)$, making scenario $\mathcal{S}_1$
infeasible for country~1.

Now, if the horizon is limited to $\tau=2$, all conditions are either
trivially satisfied or of the type $v(ij)\geqslant v(i)+v(j)$.
\end{example}

The next proposition gives some conditions about feasibility.

\begin{proposition}
\label{prop:feas} Suppose that $v$ is superadditive. The following holds:

\begin{enumerate}
\item Suppose in addition that the Markov process has all absorbing classes
which are singletons, say $S_1,\ldots,S_k$. Take a country $i\in N\setminus
(S_1\cup\cdots\cup S_k)$ (i.e., not in an absorbing class) and suppose that $%
v(\{i\})>0$. Then any scenario beginning with $\varnothing\rightarrow \{i\}$
is infeasible for $i$ for some horizon $\tau$.

\item Consider a scenario terminating at $S$ and $i\in S$. If $i$ has never
left the scenario, then the scenario is feasible for $i$.
\end{enumerate}
\end{proposition}

\begin{proof}
  \begin{enumerate}
  \item 
  Consider such an $i\in N$. As $i$ belongs to none of the absorbing classes, it
must at some time $\tau$ go out the current coalition, say $S$:
\[
\cS = \varnothing\rightarrow \underbrace{\{i\}\rightarrow\cdots\rightarrow
  S}_{i\text{ present}}\rightarrow S\setminus\{i\}\rightarrow\cdots
\]
The scenario-value on horizon $\tau$ is then
\[
\phi^\cS_i(v) = v(\{i\})-0 + v(S\setminus\{i\})-v(S)\leqslant 0<v(\{i\})
\]
by superadditivity of $v$.
\item Let us consider
a scenario $\cS$ terminating at $S$, and consider $i\in S$ such that $i$ enters
at some point, say $T$, and never leaves. We obtain
\[
\phi^\cS_i=v(T\cup \{i\})-v(T)\geqslant v(\{i\})
\]
by superadditivity of $v$.
  \end{enumerate}
\end{proof}

\begin{remark}
\begin{enumerate}
\item The first result must be taken with care, as it is possible that $i$
enters again the scenario later: $T\rightarrow T\cup \{i\}$. Even if it has
to leave again, say from set $T^{\prime }$, the sign of the process-value is
not necessarily negative: 
\begin{equation*}
\phi^\mathcal{S}_i(v) = v(\{i\})+v(S\setminus\{i\})-v(S)+v(T\cup\{i\})-v(T)
+ v(T^{\prime }\setminus\{i\})-v(T^{\prime }).
\end{equation*}
Hence, everything depends on the horizon.

\item A similar remark holds for the second result. Consider $i\in S$,
entering and leaving at some points, for example: 
\begin{equation*}
\cdots\rightarrow T\rightarrow T\cup\{i\}\rightarrow \cdots\rightarrow
T^{\prime }\rightarrow T^{\prime }\setminus\{i\}\rightarrow\cdots\rightarrow
S\setminus \{i\}\rightarrow S.
\end{equation*}
Then 
\begin{equation*}
\phi^\mathcal{S}_i=\underbrace{v(T\cup\{i\})-v(T)}_{\geqslant v(\{i\})} + 
\underbrace{v(T^{\prime }\setminus\{i\})-v(T^{\prime })}_{\leqslant
-v(\{i\})}+\underbrace{v(S)-v(S\setminus\{i\})}_{\geqslant v(\{i\})}.
\end{equation*}
Nothing can be said in general.

\item The foregoing results give some insight in the following question:
When does the scenario-value yield an imputation? The question concerns only
the countries who are active in the given scenario, as one may consider that
the other countries $i$ should simply receive their individual payoff $\pi
_{i}(e^{S})$, where $S$ is the final state of the scenario. While there is
no general answer to this question, we can say that every country entering
the scenario and not leaving it will receive a payoff at least equal to its
individual payoff.\bigskip
\end{enumerate}
\end{remark}

\section{An example of formation of an IEA}

\label{sec:example}

\subsection{The setting}

\label{sec:ieavul}

We consider the set of $n$ countries $N=\{1,\ldots ,n\}$, where $k$
countries are vulnerable to pollution (these countries form the set $K$),
whereas the rest $n-k$ countries are not (set $\bar{K}$), where $K\cup \bar{K%
}=N$. One interpretation is that vulnerable players are developed countries
whose citizens are highly sensitive to environmental issues, whereas
invulnerable players are the developing countries having to deal with
economic and social pressing issues while the environment is not a priority.
Another interpretation is that invulnerable countries face a very small
environmental damage, which we set equal to zero to simplify the analysis.
The idea of dividing the set of countries into vulnerable and non-vulnerable
players has been initiated in \cite{key-10}, and later developed in \cite%
{maza13}, \cite{supa23,supa24}, and \cite{Zhou-Mazalov}.

Referring to the general framework introduced in Section~\ref{sec:fram}, the
benefit and damage functions (see (\ref{eq:1})) are as follows: 
\begin{align*}
B_{i}(e_{i})& =\alpha _{i}e_{i}-\frac{1}{2}e_{i}^{2},\quad i\in N, \\
D_i \left( E\right) & =\left\{ 
\begin{array}{cc}
\frac{1}{2}\beta E^{2}, & \text{for }i\in K, \\ 
0, & \text{for }i\in \bar{K},%
\end{array}%
\right.
\end{align*}%
where $\alpha _{i}$ and $\beta$ are positive parameters, and $E=\sum_{i\in
N}e_{i}$ is the total pollution.

Consider the scenario where coalition $S\subseteq N$, $|S|=s$, is formed in
the game, with $\ell $ vulnerable players, $|S\cap K|=\ell $, and $s-\ell $
invulnerable players, $|S\cap \bar{K}|=s-\ell $. The players in coalition $S$
maximize their joint payoff, while each non-member maximizes its own payoff.

\begin{remark}
The (usual) assumption in the literature is that all countries are
vulnerable to pollution. This model is nested in ours. Indeed, it suffices
to set $k=n$ and $l=s$ to recover the case with only vulnerable countries.
Further, if symmetry is required, then one sets $\alpha _{i}=\alpha $ for
all $i\in N$.
\end{remark}

Assuming that the game is played noncooperatively, the following proposition
gives the Nash equilibrium strategies. We use the notation $\alpha
_{S}=\sum_{i\in S}\alpha _{i}$.

\begin{proposition}
\label{prop: IEA asym} When coalition $S\neq \varnothing $ is formed, the
players' Nash equilibrium strategies are given by 
\begin{equation}
e_{i}^{S}=%
\begin{cases}
\alpha _{i}-\frac{\ell \beta }{1+\beta (s\ell +k-\ell )}\alpha _{N}, & \text{
if }i\in S\,\text{(any player in coalition $S$)}, \\ 
\alpha _{i}-\frac{\beta }{1+\beta (s\ell +k-\ell )}\alpha _{N}, & \text{ if }%
i\in (N\setminus S)\cap K\,\text{(vulnerable player outside $S$)}, \\ 
\alpha _{i}, & \text{ if }i\in (N\setminus S)\cap \bar{K}\,\text{%
(invulnerable player outside $S$)}.%
\end{cases}
\label{eq: strategies asym}
\end{equation}%
If $S$ is formed, the total equilibrium emissions are given by 
\begin{equation}
E^{S}=\frac{1}{1+\beta (s\ell +k-\ell )}\alpha _{N},  \label{eq: emissions}
\end{equation}%
and the players' payoffs by 
\begin{equation}
\pi _{i}^{S}=%
\begin{cases}
\frac{\alpha _{i}^{2}}{2}-\frac{(\ell ^{2}\beta ^{2}+\beta )\alpha _{N}^{2}}{%
2(1+\beta (s\ell +k-\ell ))^{2}}, & \text{ if }i\in S\cap K\,\text{%
(vulnerable player in $S$)}, \\ 
\frac{\alpha _{i}^{2}}{2}-\frac{\ell ^{2}\beta ^{2}\alpha _{N}^{2}}{%
2(1+\beta (s\ell +k-\ell ))^{2}}, & \text{ if }i\in S\cap \bar{K}\,\text{%
(invulnerable player in $S$)}, \\ 
\frac{\alpha _{i}^{2}}{2}-\frac{(\beta ^{2}+\beta )\alpha _{N}^{2}}{%
2(1+\beta (s\ell +k-\ell ))^{2}}, & \text{ if }i\in (N\setminus S)\cap K\,%
\text{(vulnerable player outside $S$)}, \\ 
\frac{\alpha _{i}^{2}}{2}, & \text{ if }i\in (N\setminus S)\cap \bar{K}\,%
\text{(invulnerable player outside $S$)}.%
\end{cases}
\label{eq: payoffs}
\end{equation}
\end{proposition}

\begin{proof}
The payoff of coalition $S$ is 
\begin{equation*}
\pi_S(e_S,\{e_j\}_{j\in N\setminus S})
=\sum_{i\in S}\left(\alpha_i e_i-\frac{e_i^2}{2}\right)-\frac{1}{2}\ell\beta E^2,
\end{equation*}	
while the payoff of any individual player $j\in N\setminus S$ is
\begin{equation*}
	\pi_j(e_S,\{e_i\}_{i\in N\setminus S})=
	\begin{cases}
		\alpha_j e_j-\frac{e_j^2}{2}-\frac{1}{2}\beta E^2, & \text{ if } i\in (N\setminus S)\cap K\, \text{(vulnerable player outside $S$)},\\
	\alpha_j e_j-\frac{e_j^2}{2}, &\text{ if } i\in (N\setminus S)\cap \bar{K}\, \text{(invulnerable player outside  $S$)}.
	\end{cases}
\end{equation*}	
The first-order conditions of the optimization problems of the coalition $S$ and all individual players when identifying the Nash equilibrium result in the system of equations:
\begin{equation*}
	\begin{cases}
		\alpha_i - e_i-\ell\beta E=0, & \text{ if } i\in S\, \text{(any player from $S$)},\\
		\alpha_i - e_i-\beta E=0, & \text{ if } i\in (N\setminus S)\cap K\, \text{(vulnerable player outside $S$)},\\
		\alpha_i - e_i=0, &\text{ if } i\in (N\setminus S)\cap \bar{K}\, \text{(invulnerable player outside  $S$)}.
	\end{cases}
\end{equation*}	
Solving this system of linear equations, we obtain the strategies given by \eqref{eq: strategies asym}. After verifying the second-order conditions, we obtain that this profile of strategies admits the Nash equilibrium in the scenario when coalition $S$ is formed and all other players behave as singletons. 
\end{proof}

We can make several observations from Proposition \ref{prop: IEA asym}:

\begin{enumerate}
\item The total emission $E^{S}$ is decreasing in the number of vulnerable
countries $\ell $ in coalition $S$, the total number of vulnerable countries 
$k$, and in the size $s$ of coalition $S$.

\item The largest emissions $e_{i}^{S}$ are made by invulnerable players
outside $S$, and the smallest emissions are by any player from coalition $S$%
, independently if this player is invulnerable or not. The emissions of the
vulnerable player outside $S$ are in between.

\item Comparing the payoffs, we conclude that for both types of players,
being a singleton (outside coalition $S$) is always beneficial compared to
being a member of coalition $S$ when $S$ is fixed. However, this is not true
if we consider whether a particular player joins or leaves coalition $S$,
because in the latter case the size of the coalition changes after such
transitions.
\end{enumerate}

\begin{corollary}
In the cooperative scenario where coalition $N$ forms, the players'
emissions are given by 
\begin{equation}
e_{i}^{c}=\alpha _{i}-\frac{\beta k}{1+\beta kn}\alpha _{N},\text{ for any }%
i\in N,
\end{equation}%
and the total emissions by 
\begin{equation}
E^{c}=\frac{1}{1+\beta kn}\alpha _{N}.
\end{equation}%
The total payoff of the grand coalition $N$ is given by 
\begin{equation}
\pi _{N}^{c}=\sum_{i\in N}\frac{\alpha _{i}^{2}}{2}-\frac{\beta k}{1+\beta kn%
}\alpha _{N}^{2}.
\end{equation}
\end{corollary}

\begin{proof}
This immediately follows from Proposition \ref{prop: IEA asym} by substituting $s=n$, $\ell=k$ into equations \eqref{eq: strategies asym}, \eqref{eq: emissions}, and \eqref{eq: payoffs}.
\end{proof}

\begin{corollary}
In the fully noncooperative scenario, i.e., when all players from $N$ behave
as singletons, the players' Nash equilibrium strategies are as follows: 
\begin{equation}
e_{i}^{nc}=%
\begin{cases}
\alpha _{i}-\frac{\beta }{1+\beta k}\alpha _{N}, & \text{ if }i\in K\,\text{%
(vulnerable player)}, \\ 
\alpha _{i}, & \text{ if }i\in \bar{K}\,\text{(invulnerable player)}.%
\end{cases}%
\end{equation}%
The total emissions are given by 
\begin{equation}
E^{nc}=\frac{1}{1+\beta k}\alpha _{N},
\end{equation}%
and player $i$'s payoff by 
\begin{equation}
\pi _{i}^{nc}=%
\begin{cases}
\frac{\alpha _{i}^{2}}{2}-\frac{(\beta ^{2}+\beta )\alpha _{N}^{2}}{%
2(1+\beta k)^{2}}, & \text{ if }i\in K\,\text{(vulnerable player)}, \\ 
\frac{\alpha _{i}^{2}}{2}, & \text{ if }i\in \bar{K}\,\text{(invulnerable
player)}.%
\end{cases}%
\end{equation}
\end{corollary}

\begin{proof}
This follows directly from Proposition \ref{prop: IEA asym} by substituting $s=1$, $\ell=1$ for a vulnerable player and $s=1$, $\ell=0$ for an invulnerable player into the equations \eqref{eq: strategies asym} and \eqref{eq: payoffs}. Then we can find the total emissions by summarizing the equilibrium individual emissions over the set of players $N$.
\end{proof}

As it has been repeatedly shown in the literature (see \cite%
{Jorgensen2010433} for a survey), total emissions are lower under (full)
cooperation than under a noncooperative mode of play.

We examine when players (vulnerable and invulnerable) have incentives to
join or leave coalition $S$. To do this, we compute $\delta_i^-(S)$ for any
player $i\in S$ and $\delta_i^+(S)$ for any player $i\notin S$, depending on
whether player $i$ is vulnerable, $i\in K$, or not, $i\in \bar{K}$:

\begin{enumerate}
\item \textit{Invulnerable player }$i$\textit{\ in coalition} $S$, $i\in
S\cap \bar{K}$, $S\neq \{i\}$. Taking into account the expression of player $%
i$'s payoff given in Proposition \ref{prop: IEA asym}, by formula %
\eqref{eq:d-} we obtain: 
\begin{equation}
\delta _{i}^{-}(S)=-\frac{\ell ^{2}\beta ^{2}\alpha _{N}^{2}}{2(1+\beta
(s\ell +k-\ell ))^{2}}+c,   \label{eq: inv-}
\end{equation}
which can be of any sign depending on parameters of the game.

\item \textit{Invulnerable player outside coalition }$S$, $i\in (N\setminus
S)\cap \bar{K}$, $S\neq \varnothing ,S\subseteq N\setminus \{i\}$. By
formula \eqref{eq:d+} we find that 
\begin{equation}
\delta _{i}^{+}(S)=-\frac{\ell ^{2}\beta ^{2}\alpha _{N}^{2}}{2(1+\beta
((s+1)\ell +k-\ell ))^{2}}<0.  \label{eq: inv+}
\end{equation}

\item \textit{Vulnerable player }$i$\textit{\ in coalition }$S$, $i\in S\cap
K$, $S\neq \{i\}$. Taking into account the vulnerable player $i$'s payoff
defined in Proposition \ref{prop: IEA asym}, by formula \eqref{eq:d-} we
obtain: 
\begin{equation}
\delta _{i}^{-}(S)=\left( \frac{(\beta ^{2}+\beta )}{2(1+\beta ((s-1)(\ell
-1)+k-\ell +1))^{2}}-\frac{(\ell ^{2}\beta ^{2}+\beta )}{2(1+\beta (s\ell
+k-\ell ))^{2}}\right) \alpha _{N}^{2}+c,  \label{eq: vul-}
\end{equation}%
which can be of any sign depending on parameters of the game.

\item \textit{Vulnerable player outside coalition }$S$, $i\in (N\setminus
S)\cap K$, $S\neq \varnothing ,S\subseteq N\setminus \{i\}$. By formula %
\eqref{eq:d+} we find that 
\begin{equation}
\delta _{i}^{+}(S)=\left( \frac{(\beta ^{2}+\beta )}{2(1+\beta (s\ell
+k-\ell ))^{2}}-\frac{((\ell +1)^{2}\beta ^{2}+\beta )}{2(1+\beta
((s+1)(\ell +1)+k-\ell -1))^{2}}\right) \alpha _{N}^{2},  \label{eq: vul+}
\end{equation}%
which again can be of any sign depending on parameters of the game.

\item \textit{Invulnerable player }$i$\textit{\ joining }$\varnothing $%
\textit{\ or leaving }$\{i\}$\textit{.} According to our convention and
formulas (\ref{eq:d+}), (\ref{eq:d-}), 
\begin{equation*}
\delta _{i}^{+}(\varnothing )=\delta _{i}^{-}(\{i\})=\frac{\alpha _{i}^{2}}{2%
}>0.
\end{equation*}

\item \textit{Vulnerable player }$i$\textit{\ joining }$\varnothing $ 
\textit{or leaving }$\{i\}$. According to our convention and formulas (\ref%
{eq:d+}), (\ref{eq:d-}), 
\begin{equation*}
\delta _{i}^{+}(\varnothing )=\delta _{i}^{-}(\{i\})=\frac{\alpha _{i}^{2}}{2%
}-\frac{(\beta ^{2}+\beta )\alpha _{N}^{2}}{2(1+\beta k)^{2}},
\end{equation*}%
which can be of any sign depending on parameters of the game.
\end{enumerate}

Based on the above cases of player's possible transitions, we can draw the
following conclusions:

\begin{enumerate}
\item If an invulnerable country belongs to $S$, $|S|>1$, it has an
incentive to leave $S$ when 
\begin{equation*}
c<\frac{\ell ^{2}\beta ^{2}\alpha _{N}^{2}}{2(1+\beta (s\ell +k-\ell ))^{2}}.
\end{equation*}

\item If an invulnerable country does not belong to $S$, $S\neq \varnothing$%
, it never has an incentive to join~$S$.

\item The sign of $\delta _{i}^{+}(S)$, $S\neq \varnothing$, is independent
of any values $\alpha _{j}$, $j\in N$, but it is dependent of them for $%
\delta _{i}^{+}(\varnothing )$ and $\delta _{i}^{-}(S)$ for any $S$.

\item If a vulnerable country $i$ belongs to $S\neq \{i\}$, it has an
incentive to leave $S$ if 
\begin{equation*}
c<\left( \frac{(\ell ^{2}\beta ^{2}+\beta )}{2(1+\beta (s\ell +k-\ell ))^{2}}%
-\frac{(\beta ^{2}+\beta )}{2(1+\beta ((s-1)(\ell -1)+k-\ell +1))^{2}}%
\right) \alpha _{N}^{2}.
\end{equation*}%
If $c=0$, then this inequality is reduced to the following one: 
\begin{align}
& \beta ^{2}(k^{2}-2k\ell -3\ell ^{2}-4k\ell ^{2}-k^{2}\ell ^{2}+8\ell
^{3}+4k\ell ^{3}-4\ell ^{4}+2k\ell s+2\ell ^{2}s+2k\ell ^{2}s  \notag \\
& -8\ell ^{3}s-2k\ell ^{3}s+4\ell ^{4}s+2\ell ^{3}s^{2}-\ell ^{4}s^{2}) 
\notag \\
& +\beta (-4-2k+6\ell +2k\ell -7\ell ^{2}-2k\ell ^{2}+4\ell
^{3}+4s+2ks-6\ell s+4\ell ^{2}s-2\ell ^{3}s-s^{2}+2\ell s^{2})  \notag \\
& +(-3+2\ell -\ell ^{2}+2s)<0.  \label{eq: d-<0}
\end{align}%
Since inequality \eqref{eq: d-<0} does not depend on $\alpha _{j}$, $j\neq i$%
, the vulnerable players are symmetric when deciding on leaving $S$ when $c=0
$.

\item If a vulnerable country does not belong to $S\neq \varnothing $, it
has an incentive to join $S$ if 
\begin{align}
& \beta ^{2}(2k\ell -2k^{2}\ell -\ell ^{2}+4k\ell ^{2}-k^{2}\ell ^{2}-2\ell
^{3}+2k\ell ^{3}-\ell ^{4}+2ks+2\ell ^{2}s-4k\ell ^{2}s  \notag \\
& +4\ell ^{3}s-2k\ell ^{3}s+2\ell ^{4}s+s^{2}+2\ell s^{2}-2\ell
^{3}s^{2}-\ell ^{4}s^{2})  \notag \\
& +\beta (2\ell -2k\ell +3\ell ^{2}-2k\ell ^{2}+2\ell ^{3}+2s+2ks-2\ell
^{2}s-2\ell ^{3}s+s^{2}+2\ell s^{2})  \notag \\
& +(2s-\ell ^{2})>0.  \label{eq: d+>0}
\end{align}%
As this expression does not depend on $\alpha _{j}$, $j\neq i$, the
vulnerable players are symmetric when deciding on joining $S$.

\item At the origin of the scenario ($S_{0}=\varnothing $), any invulnerable
player can join $\varnothing $ (i.e., can start the coalition formation
process). A vulnerable player $i$ starts the coalition formation process iff 
\begin{equation*}
\alpha _{i}^{2}(1+\beta k)^2>\beta (1+\beta )\alpha _{N}^{2}.
\end{equation*}

\item We analyze in detail the case when there is no cost for leaving
coalition $S$, i.e. $c=0$. After the first transition ($|S|=1$), only a
vulnerable player will join. A vulnerable player has an incentive to join
coalition $S$ consisting of another vulnerable player iff condition %
\eqref{eq: d+>0} holds true for $s=\ell =1$, which is equivalent to the
inequality 
\begin{equation*}
\beta ^{2}(4+4k-3k^{2})+\beta (8-2k)+1>0.
\end{equation*}

Any vulnerable player belonging to coalition $S$ containing two vulnerable
players has no incentive to leave the coalition iff 
\begin{equation*}
\beta^2 (4 + 4 k - 3 k^2)+ \beta (8 - 2 k)+ 1>0. 
\end{equation*}
For $k=2$ this inequality is equivalent to $\beta> -\frac{1}{4}$, and it is
always satisfied.

The third vulnerable player has an incentive to join iff 
\begin{equation*}
(16-16k)+\beta (-48k-16k^{2})>0\Leftrightarrow \beta <-\frac{k-1}{k(k+3)},
\end{equation*}%
which is never satisfied.

The fourth vulnerable player has an incentive to join iff 
\begin{equation*}
-6+\beta (-120-36k)+\beta ^{2}(-864-336k-30k^{2})>0,
\end{equation*}%
which is never satisfied.

\item As a result, for the case when $c=0$ we obtain that the absorbing
state may be only a coalition consisting of two vulnerable players but never
more, and it also depends on $\beta $ and $k$ ($\beta >0$ for $k=2$; $\beta
\in (0,(1+2\sqrt{3})/11)$ for $k=3$; $\beta \in (0,1/2\sqrt{7})$ for $k=4$; $%
\beta \in (0,(-1+2\sqrt{13})/51)$ for $k=5$, etc.)

\item If the game $v$ is superadditive and $c=0$, by combining the above
results with Proposition~\ref{prop:feas} (Item 1), we find that any scenario
starting with $\varnothing\rightarrow \{i\}$, where $i$ is invulnerable, is
infeasible.
\end{enumerate}

\subsection{Numerical example}

We apply the above results to an example with four countries $N=\{1,2,3,4\}$%
, where countries 1, 2 and 3 are vulnerable to pollution, whereas country 4
is not. To go further in the analysis, we choose the following parameter
values:%
\begin{equation*}
\alpha _{1}=6,\text{ \ }\alpha _{2}=5,\text{ \ }\alpha _{3}=4,\text{ \ }
\alpha _{4}=3,\text{ \ }\beta =0.05.
\end{equation*}
Table~\ref{tab:1} reports the results of emissions, for all possible
coalitions given in Proposition \ref{prop: IEA asym}.

\begin{table}[htb]
\caption{Strategies (level of emissions)}
\label{tab:1}
\begin{center}
{\footnotesize \ \hspace*{-1.5cm}%
\begin{tabular}{c|ccccc}
\hline
Coalition & $e_{1}$ & $e_{2}$ & $e_{3}$ & $e_{4}$ & $E$ \\ \hline
$\{1\}$ & $5.217$ & $4.217$ & $3.217$ & $3$ & $15.652$ \\ 
$\{2\}$ & $5.217$ & $4.217$ & $3.217$ & $3$ & $15.652$ \\ 
$\{3\}$ & $5.217$ & $4.217$ & $3.217$ & $3$ & $15.652$ \\ 
$\{4\}$ & $5.217$ & $4.217$ & $3.217$ & $3$ & $15.652$ \\ 
$\{1,2\}$ & $4.560$ & $3.560$ & $3.280$ & $3$ & $14.400$ \\ 
$\{1,3\}$ & $4.560$ & $4.280$ & $2.560$ & $3$ & $14.400$ \\ 
$\{1,4\}$ & $5.250$ & $4.250$ & $3.250$ & $2.250$ & $15.000$ \\ 
$\{2,3\}$ & $5.280$ & $3.560$ & $2.560$ & $3$ & $14.400$ \\ 
$\{2,4\}$ & $5.250$ & $4.250$ & $3.250$ & $2.250$ & $15.000$ \\ 
$\{3,4\}$ & $5.250$ & $4.250$ & $3.250$ & $2.250$ & $15.000$ \\ 
$\{1,2,3\}$ & $4.138$ & $3.138$ & $2.138$ & $3$ & $12.414$ \\ 
$\{1,2,4\}$ & $4.667$ & $3.667$ & $3.333$ & $1.667$ & $13.333$ \\ 
$\{1,3,4\}$ & $4.667$ & $4.333$ & $2.667$ & $1.667$ & $13.333$ \\ 
$\{2,3,4\}$ & $5.333$ & $3.667$ & $2.667$ & $1.667$ & $13.333$ \\ 
$\{1,2,3,4\}$ & $4.3125$ & $3.3125$ & $2.3125$ & $1.3125$ & $11.250$ \\ 
\hline
\end{tabular}
}
\end{center}
\end{table}

In this example, we have two asymmetries, one related to the environmental
damage cost and the other in the different values of $\alpha$. A higher
value of $\alpha$ means that player has a more efficient economy. Table~\ref%
{tab:payoff} gives the corresponding payoffs and the cooperative game $v$. 
\begin{table}[h!]
\caption{Payoffs}
\label{tab:payoff}{\footnotesize \ }
\par
\begin{center}
{\footnotesize \ 
\begin{tabular}{c|c|cccc}
\hline
Coalition $S$ & $v\left( S\right) $ & \multicolumn{4}{c}{Players' Payoffs}
\\ \hline
&  & $\pi _{1}(e^{S})$ & $\pi _{2}(e^{S})$ & $\pi _{3}(e^{S})$ & $\pi
_{4}(e^{S})$ \\ \hline
$\{1\}$ & $11.569$ & $11.569$ & $6.069$ & $1.569$ & $4.500$ \\ 
$\{2\}$ & $6.069$ & $11.569$ & $6.069$ & $1.569$ & $4.500$ \\ 
$\{3\}$ & $1.569$ & $11.569$ & $6.069$ & $1.569$ & $4.500$ \\ 
$\{4\}$ & $4.500$ & $11.569$ & $6.069$ & $1.569$ & $4.500$ \\ 
$\{1,2\}$ & $18.058$ & $11.779$ & $6.279$ & $2.557$ & $4.500$ \\ 
$\{1,3\}$ & $13.558$ & $11.779$ & $7.057$ & $1.779$ & $4.500$ \\ 
$\{1,4\}$ & $16.313$ & $12.094$ & $6.594$ & $2.094$ & $4.219$ \\ 
$\{2,3\}$ & $8.058$ & $12.557$ & $6.279$ & $1.779$ & $4.500$ \\ 
$\{2,4\}$ & $10.813$ & $12.094$ & $6.594$ & $2.094$ & $4.219$ \\ 
$\{3,4\}$ & $6.313$ & $12.094$ & $6.594$ & $2.094$ & $4.219$ \\ 
$\{1,2,3\}$ & $21.741$ & $12.414$ & $6.914$ & $2.414$ & $4.500$ \\ 
$\{1,2,4\}$ & $23.444$ & $12.667$ & $7.167$ & $3.333$ & $3.611$ \\ 
$\{1,3,4\}$ & $18.944$ & $12.667$ & $7.833$ & $2.667$ & $3.611$ \\ 
$\{2,3,4\}$ & $13.444$ & $13.333$ & $7.167$ & $2.667$ & $3.611$ \\ 
$\{1,2,3,4\}$ & $27.813$ & $13.412$ & $7.912$ & $3.412$ & $3.076$ \\ \hline
\end{tabular}
}
\end{center}
\end{table}

\begin{remark}
\label{rem: super 2} We can easily notice that the characteristic function
with values given in Table~\ref{tab:payoff} is superadditive. However, this
is not a general result. For example, if $\beta =0.44$ and we keep all other
parameters the same, then the game is not superadditive. We observe that $%
v(\{1\})=-1.07$, $v(\{2\})=-6.57$, and $v(\{1,2\})=-7.92$. Since $%
v(\{1\})+v(\{2\})>v(\{1,2\})$, for vulnerable countries 1 and 2 it is not
profitable to form a coalition.\footnote{%
In this example, the equilibrium emissions are positive for all the players,
so we have an interior solution with negative resulting payoffs.}
\end{remark}

\subsubsection{Without cost of leaving ($c=0$)}

\label{sec:c=0} We proceed state by state, starting with the empty
coalition. We omit the details of computation, solely based on (\ref{eq:d-})
and (\ref{eq:d+}), and in Figure \ref{fig: trans graph} we represent the
transition graph (absorbing states are in boldface), together with all
probabilities of transition in Table \ref{tab: trans prob}. We recall that
these probabilities are proportional to $\delta^+_i(S)$ or $|\delta^-_i(S)|$.

\begin{figure}[h!]
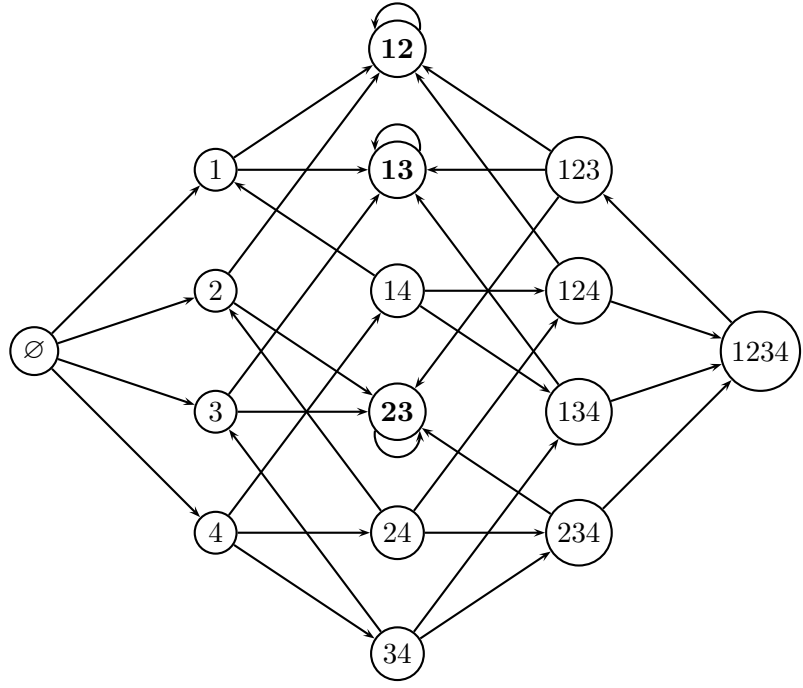

\begin{center}
\psset{unit=0.8cm} \pspicture(-0.2,-1)(12,11) \cnodeput(0,5){0}{$\varnothing$%
} \cnodeput(3,8){1}{$1$} \cnodeput(3,6){2}{$2$} \cnodeput(3,4){3}{$3$} %
\cnodeput(3,2){4}{$4$} \cnodeput(6,10){12}{$\mathbf{12}$} \cnodeput(6,8){13}{%
$\mathbf{13}$} \cnodeput(6,6){14}{$14$} \cnodeput(6,4){23}{$\mathbf{23}$} %
\cnodeput(6,2){24}{$24$} \cnodeput(6,0){34}{$34$} \cnodeput(9,8){123}{$123$} %
\cnodeput(9,6){124}{$124$} \cnodeput(9,4){134}{$134$} \cnodeput(9,2){234}{$%
234$} \cnodeput(12,5){1234}{$1234$} \ncarc[arcangle=0]{->}{0}{1} 
\ncarc[arcangle=0]{->}{0}{2} \ncarc[arcangle=0]{->}{0}{3} %
\ncarc[arcangle=0]{->}{0}{4} \ncarc[arcangle=0]{->}{1}{12} %
\ncarc[arcangle=0]{->}{1}{13} \ncarc[arcangle=0]{->}{2}{12} %
\ncarc[arcangle=0]{->}{2}{23} \ncarc[arcangle=0]{->}{3}{13} %
\ncarc[arcangle=0]{->}{3}{23} \ncarc[arcangle=0]{->}{4}{14} %
\ncarc[arcangle=0]{->}{4}{24} \ncarc[arcangle=0]{->}{4}{34} %
\ncarc[arcangle=0]{->}{14}{1} \ncarc[arcangle=0]{->}{14}{124} %
\ncarc[arcangle=0]{->}{14}{134} \ncarc[arcangle=0]{->}{24}{2} %
\ncarc[arcangle=0]{->}{24}{124} \ncarc[arcangle=0]{->}{24}{234} %
\ncarc[arcangle=0]{->}{34}{3} \ncarc[arcangle=0]{->}{34}{134} %
\ncarc[arcangle=0]{->}{34}{234} \ncarc[arcangle=0]{->}{123}{12} %
\ncarc[arcangle=0]{->}{123}{13} \ncarc[arcangle=0]{->}{123}{23} %
\ncarc[arcangle=0]{->}{124}{12} \ncarc[arcangle=0]{->}{124}{1234} %
\ncarc[arcangle=0]{->}{134}{13} \ncarc[arcangle=0]{->}{134}{1234} %
\ncarc[arcangle=0]{->}{234}{23} \ncarc[arcangle=0]{->}{234}{1234} %
\ncarc[arcangle=0]{->}{1234}{123} \nccircle{->}{12}{0.3cm} %
\nccircle{->}{13}{0.3cm} \nccircle{->}{23}{-0.3cm} \endpspicture
\end{center}
\caption{Transition graph with $c=0$}
\label{fig: trans graph}
\end{figure}

\begin{table}[h!]
\caption{Transition probabilities with $c=0$}
\label{tab: trans prob}
\begin{center}
\begin{tabular}{|c|c||c|c|}
\hline
$p_{\varnothing,1}$ & 0.488 & $p_{24,2}$ & 0.296 \\ 
$p_{\varnothing,2}$ & 0.256 & $p_{24,124}$ & 0.352 \\ 
$p_{\varnothing,3}$ & 0.066 & $p_{24,234}$ & 0.352 \\ 
$p_{\varnothing,4}$ & 0.190 & $p_{34,3}$ & 0.296 \\ 
$p_{1,12}$ & 0.5 & $p_{34,134}$ & 0.352 \\ 
$p_{1,13}$ & 0.5 & $p_{34,234}$ & 0.352 \\ 
$p_{2,12}$ & 0.5 & $p_{123,12}$ & 0.333 \\ 
$p_{2,23}$ & 0.5 & $p_{123,13}$ & 0.333 \\ 
$p_{3,13}$ & 0.5 & $p_{123,23}$ & 0.333 \\ 
$p_{3,23}$ & 0.5 & $p_{124,12}$ & 0.918 \\ 
$p_{4,14}$ & 0.333 & $p_{124,1234}$ & 0.082 \\ 
$p_{4,24}$ & 0.333 & $p_{134,13}$ & 0.918 \\ 
$p_{4,34}$ & 0.333 & $p_{134,1234}$ & 0.082 \\ 
$p_{14,1}$ & 0.296 & $p_{234,23}$ & 0.918 \\ 
$p_{14,124}$ & 0.352 & $p_{234,1234}$ & 0.082 \\ 
$p_{14,134}$ & 0.352 & $p_{1234,123}$ & 1 \\ \hline
\end{tabular}%
\end{center}
\end{table}

We list now all possible scenarios starting from $\varnothing$ and compute
probability of their occurrence by equation \eqref{prob occurrence} (product
of all transition probabilities along the scenario, assuming the probability
to start from $\varnothing$ is 1).

\begin{itemize}
\item $\mathcal{S}_1$: $\varnothing\rightarrow 1\rightarrow 12$, ${\mathbb{P}%
}(\mathcal{S}_1) = 0.244$

\item $\mathcal{S}_2$: $\varnothing\rightarrow 1\rightarrow 13$, ${\mathbb{P}%
}(\mathcal{S}_2) = 0.244$

\item $\mathcal{S}_3$: $\varnothing\rightarrow 2\rightarrow 12$, ${\mathbb{P}%
}(\mathcal{S}_3) = 0.128$

\item $\mathcal{S}_4$: $\varnothing\rightarrow 2\rightarrow 23$, ${\mathbb{P}%
}(\mathcal{S}_4) = 0.128$

\item $\mathcal{S}_5$: $\varnothing\rightarrow 3\rightarrow 13$, ${\mathbb{P}%
}(\mathcal{S}_5) = 0.033$

\item $\mathcal{S}_6$: $\varnothing\rightarrow 3\rightarrow 23$, ${\mathbb{P}%
}(\mathcal{S}_6) = 0.033$

\item $\mathcal{S}_7$: $\varnothing\rightarrow 4\rightarrow 14\rightarrow
124\rightarrow 1234\rightarrow 123\rightarrow 12$, ${\mathbb{P}}(\mathcal{S}%
_7) = 0.0006$

\item $\mathcal{S}_8$: $\varnothing\rightarrow 4\rightarrow 14\rightarrow
124\rightarrow 1234\rightarrow 123\rightarrow 13$, ${\mathbb{P}}(\mathcal{S}%
_8) = 0.0006$

\item $\mathcal{S}_9$: $\varnothing\rightarrow 4\rightarrow 14\rightarrow
124\rightarrow 1234\rightarrow 123\rightarrow 23$, ${\mathbb{P}}(\mathcal{S}%
_9) = 0.0006$

\item $\mathcal{S}_{10}$: $\varnothing\rightarrow 4\rightarrow 14\rightarrow
124\rightarrow 12$, ${\mathbb{P}}(\mathcal{S}_{10}) = 0.0205$

\item $\mathcal{S}_{11}$: $\varnothing\rightarrow 4\rightarrow 14\rightarrow
134\rightarrow 13$, ${\mathbb{P}}(\mathcal{S}_{11}) = 0.0205$

\item $\mathcal{S}_{12}$: $\varnothing\rightarrow 4\rightarrow 14\rightarrow
134\rightarrow 1234\rightarrow 123\rightarrow 12$, ${\mathbb{P}}(\mathcal{S}%
_{12}) = 0.0006$

\item $\mathcal{S}_{13}$: $\varnothing\rightarrow 4\rightarrow 14\rightarrow
134\rightarrow 1234\rightarrow 123\rightarrow 13$, ${\mathbb{P}}(\mathcal{S}%
_{13}) = 0.0006$

\item $\mathcal{S}_{14}$: $\varnothing\rightarrow 4\rightarrow 14\rightarrow
134\rightarrow 1234\rightarrow 123\rightarrow 23$, ${\mathbb{P}}(\mathcal{S}%
_{14}) = 0.0006$

\item $\mathcal{S}_{15}$: $\varnothing\rightarrow 4\rightarrow 14\rightarrow
1\rightarrow 12$, ${\mathbb{P}}(\mathcal{S}_{15}) = 0.0094$

\item $\mathcal{S}_{16}$: $\varnothing\rightarrow 4\rightarrow 14\rightarrow
1\rightarrow 13$, ${\mathbb{P}}(\mathcal{S}_{16}) = 0.0094$

\item $\mathcal{S}_{17}$: $\varnothing\rightarrow 4\rightarrow 24\rightarrow
2 \rightarrow 12$, ${\mathbb{P}}(\mathcal{S}_{17}) = 0.0094$

\item $\mathcal{S}_{18}$: $\varnothing\rightarrow 4\rightarrow 24\rightarrow
2 \rightarrow 23$, ${\mathbb{P}}(\mathcal{S}_{18}) = 0.0094$

\item $\mathcal{S}_{19}$: $\varnothing\rightarrow 4\rightarrow 24\rightarrow
124\rightarrow 12$, ${\mathbb{P}}(\mathcal{S}_{19}) = 0.0205$

\item $\mathcal{S}_{20}$: $\varnothing\rightarrow 4\rightarrow 24\rightarrow
124\rightarrow 1234\rightarrow 12$, ${\mathbb{P}}(\mathcal{S}_{20}) = 0.0006 
$

\item $\mathcal{S}_{21}$: $\varnothing\rightarrow 4\rightarrow 24\rightarrow
124\rightarrow 1234\rightarrow 13$, ${\mathbb{P}}(\mathcal{S}_{21}) = 0.0006 
$

\item $\mathcal{S}_{22}$: $\varnothing\rightarrow 4\rightarrow 24\rightarrow
124\rightarrow 1234\rightarrow 23$, ${\mathbb{P}}(\mathcal{S}_{22}) = 0.0006 
$

\item $\mathcal{S}_{23}$: $\varnothing\rightarrow 4\rightarrow 24\rightarrow
234\rightarrow 23$, ${\mathbb{P}}(\mathcal{S}_{23}) = 0.0205$

\item $\mathcal{S}_{24}$: $\varnothing\rightarrow 4\rightarrow 24\rightarrow
234\rightarrow 1234\rightarrow 12$, ${\mathbb{P}}(\mathcal{S}_{24}) = 0.0006 
$

\item $\mathcal{S}_{25}$: $\varnothing\rightarrow 4\rightarrow 24\rightarrow
234\rightarrow 1234\rightarrow 13$, ${\mathbb{P}}(\mathcal{S}_{25}) = 0.0006 
$

\item $\mathcal{S}_{26}$: $\varnothing\rightarrow 4\rightarrow 24\rightarrow
234\rightarrow 1234\rightarrow 23$, ${\mathbb{P}}(\mathcal{S}_{26}) = 0.0006 
$

\item $\mathcal{S}_{27}$: $\varnothing\rightarrow 4\rightarrow 34\rightarrow
3\rightarrow 13$, ${\mathbb{P}}(\mathcal{S}_{27}) = 0.0094$

\item $\mathcal{S}_{28}$: $\varnothing\rightarrow 4\rightarrow 34\rightarrow
3\rightarrow 23$, ${\mathbb{P}}(\mathcal{S}_{28}) = 0.0094$

\item $\mathcal{S}_{29}$: $\varnothing\rightarrow 4\rightarrow 34\rightarrow
134\rightarrow 13$, ${\mathbb{P}}(\mathcal{S}_{29}) = 0.0205$

\item $\mathcal{S}_{30}$: $\varnothing\rightarrow 4\rightarrow 34\rightarrow
134\rightarrow 1234 \rightarrow 123 \rightarrow 12$, ${\mathbb{P}}(\mathcal{S%
}_{30}) = 0.0006 $

\item $\mathcal{S}_{31}$: $\varnothing\rightarrow 4\rightarrow 34\rightarrow
134\rightarrow 1234 \rightarrow 123 \rightarrow 13$, ${\mathbb{P}}(\mathcal{S%
}_{31}) = 0.0006 $

\item $\mathcal{S}_{32}$: $\varnothing\rightarrow 4\rightarrow 34\rightarrow
134\rightarrow 1234 \rightarrow 123 \rightarrow 23$, ${\mathbb{P}}(\mathcal{S%
}_{32}) = 0.0006 $

\item $\mathcal{S}_{33}$: $\varnothing\rightarrow 4\rightarrow 34\rightarrow
234\rightarrow 23$, ${\mathbb{P}}(\mathcal{S}_{33}) = 0.0205$

\item $\mathcal{S}_{34}$: $\varnothing\rightarrow 4\rightarrow 34\rightarrow
234\rightarrow 1234 \rightarrow 123 \rightarrow 12$, ${\mathbb{P}}(\mathcal{S%
}_{34}) = 0.0006 $

\item $\mathcal{S}_{35}$: $\varnothing\rightarrow 4\rightarrow 34\rightarrow
234\rightarrow 1234 \rightarrow 123 \rightarrow 13$, ${\mathbb{P}}(\mathcal{S%
}_{35}) = 0.0006 $

\item $\mathcal{S}_{36}$: $\varnothing\rightarrow 4\rightarrow 34\rightarrow
234\rightarrow 1234 \rightarrow 123 \rightarrow 23$, ${\mathbb{P}}(\mathcal{S%
}_{36}) = 0.0006 $
\end{itemize}

Let us assume that the countries are farsighted. 
Observe that the game $v$ is superadditive and that the absorbing classes
are singletons (states 12, 13 and 23), as expected from our analysis in
Section~\ref{sec:ieavul}. Then, Proposition~\ref{prop:feas} (Item 1) can be
applied, and we find that all scenarios where country 4 is entering first
(scenarios $\mathcal{S}_{7}$ to $\mathcal{S}_{36}$) are infeasible. Also,
from Proposition~\ref{prop:feas} (Item 2), the 6 remaining scenarios are
feasible.

For these 6 scenarios, the scenario-values are as follows:

\begin{itemize}
\item Scenario $\mathcal{S}_{1}$: $\phi _{1}^{\mathcal{S}_{1}}=v(1)-v(%
\varnothing )=11.569$, $\phi _{2}^{\mathcal{S}_{1}}=v(12)-v(1)=6.489$. The
two countries $3$ and $4$ are inactive, hence $\phi _{3}^{\mathcal{S}%
_{1}}=\phi _{4}^{\mathcal{S}_{1}}=0$. However, they receive their individual
payoff under the formation of coalition $12$, i.e., $\pi _{3}(e^{12})=2.557$
and $\pi _{4}(e^{12})=4.500$, respectively.

\item Scenario $\mathcal{S}_{2}$: $\phi _{1}^{\mathcal{S}_{2}}=v(1)-v(%
\varnothing)=11.569$, $\phi _{3}^{\mathcal{S}_{2}}=v(13)-v(1)=1.989$.
Countries $2$ and $4$ are inactive, and receive $\pi _{2}(e^{13})=7.057$ and 
$\pi _{4}(e^{13})=4.500$, respectively.

\item Scenario $\mathcal{S}_{3}$: $\phi _{1}^{\mathcal{S}%
_{3}}=v(12)-v(2)=11.989$, $\phi _{2}^{\mathcal{S}_{3}}=v(2)-v(\varnothing
)=6.069$. Countries $3$ and $4$ are inactive, and receive $\pi
_{3}(e^{12})=2.557$ and $\pi _{4}(e^{12})=4.500$, respectively.

\item Scenario $\mathcal{S}_{4}$: $\phi _{2}^{\mathcal{S}_{4}}=v(2)-v(%
\varnothing )=6.069$, $\phi _{3}^{\mathcal{S}_{4}}=v(23)-v(2)=1.989$.
Countries $1$ and $4$ are inactive, and receive $\pi _{1}(e^{23})=12.557$
and $\pi _{4}(e^{23})=4.500$, respectively.

\item Scenario $\mathcal{S}_{5}$: $\phi _{1}^{\mathcal{S}%
_{5}}=v(13)-v(3)=11.989$, $\phi _{3}^{\mathcal{S}_{5}}=v(3)-v(\varnothing
)=1.569$. Countries $2$ and $4$ are inactive, and receive $\pi
_{2}(e^{13})=7.057$ and $\pi _{4}(e^{13})=4.500$, respectively.

\item Scenario $\mathcal{S}_{6}$: $\phi _{2}^{\mathcal{S}%
_{6}}=v(23)-v(3)=6.489$, $\phi _{3}^{\mathcal{S}_{6}}=v(3)-v(\varnothing
)=1.569$. Countries $1$ and $4$ are inactive, and receive $\pi
_{1}(e^{23})=12.557$ and $\pi _{4}(e^{23})=4.500$, respectively.
\end{itemize}

To compute the process-value, we must consider that scenarios $7$ to $36$
will not realize, therefore, we have to renormalize the probabilities of
scenarios $1$ to $6$ to obtain 
\begin{equation*}
{\mathbb{P}}(\mathcal{S}_{1})={\mathbb{P}}(\mathcal{S}_{2})=0.301,\ {\mathbb{%
P}}(\mathcal{S}_{3})={\mathbb{P}}(\mathcal{S}_{4})=0.158,\ {\mathbb{P}}(%
\mathcal{S}_{5})={\mathbb{P}}(\mathcal{S}_{6})=0.041,
\end{equation*}%
which yields the following values: 
\begin{equation}
\phi _{1}=9.3468,\ \phi _{2}=4.1351,\ \phi _{3}=1.0407,\ \phi _{4}=0. 
\label{eq:pro}
\end{equation}


\paragraph{Interpretation of the results and discussion:}

If a scenario is observed, say $\mathcal{S}_{1}$, then $\phi ^{\mathcal{S}%
_{1}}$ should be used for the payment. Recall that if a scenario $\mathcal{S}
$ terminates at coalition $S$, the sum of the values $\sum_{i\in N}\phi
_{i}^{\mathcal{S}}=v(S)$. Hence, the payoff is a sharing of $v(S)$, not
among players in $S$, but among \textit{all players that have been active in
the scenario}. In our case, for all scenarios $\mathcal{S}_{1}$ to $\mathcal{%
S}_{6}$, this amounts to the same. Recall that the inactive players receive
their individual payoff under the formation of coalition $S$.

Now suppose that no scenario has been observed, but one sees only the final
result, i.e., the coalitions that are absorbing states, in our case, $12,13$%
, and $23$. Then, it makes sense to take the expected value over all
scenarios leading to that absorbing state. In our example, taking $S=12$ as
absorbing state, only scenarios $\mathcal{S}_{1}$ and $\mathcal{S}_{3}$ lead
to $12$. The probabilities of their occurrence are $0.244$ and $0.128$,
respectively, so after normalization, this yields $0.656$ and $0.344$.
Taking the expected value yields the following payoff: 
\begin{equation*}
\phi _{1}=11.713,\ \phi _{2}=6.344,\text{ }\phi _{3}=\phi _{4}=0.
\end{equation*}%
Countries $3$ and $4$ do not participate to the scenarios and receive $\pi
_{3}(e^{12})$ and $\pi _{4}(e^{12})$, respectively.

Now taking $S=13$, scenarios $\mathcal{S}_{2}$ and $\mathcal{S}_{5}$ are
involved. After normalization, the probabilities of their occurrence are $%
0.881$ and $0.119$, respectively. This yields: 
\begin{equation*}
\phi _{1}=11.619,\ \phi _{3}=1.939,\ \phi _{2}=\phi _{4}=0.
\end{equation*}%
Finally, for $S=23$, scenarios $\mathcal{S}_{4}$ and $\mathcal{S}_{6}$ are
involved, with normalized probabilities $0.795$ and $0.205$, respectively.
This yields: 
\begin{equation*}
\phi _{2}=6.155,\ \phi _{3}=1.903,\ \phi _{1}=\phi _{4}=0.
\end{equation*}

Taking the process-value does not make so much sense, because this would
mean that no final coalition (absorbing state) will be observed, in which
case the value does not represent a sharing of something (no property of
efficiency for the process-value). Moreover, there are countries active in
some scenarios and inactive in other ones, and the obtained payoff (\ref%
{eq:pro}) is not an imputation.

An important question is indeed whether the obtained payoff vector is an
imputation, i.e., $\phi _{i}\geqslant v(\{i\})$ for every country $i$. Let
us recall the values for $v(\{i\})$: 
\begin{equation*}
v(\{1\})=11.569,\ v(\{2\})=6.069,\ v(\{3\})=1.569.
\end{equation*}%
Comparing these values with the above payoffs for \textit{active} players,
we see that in all cases, the obtained payoffs are imputations when
restricted to active players. This is due to the definition of a feasible
scenario, and since we have taken an average over feasible scenarios, the
obtained payoff is an imputation. For the inactive country $4$, its
individual payoff is always equal to $4.500$, which is exactly $v(\{4\})$.

Finally, let us remark that the payoffs obtained by our method radically
differ from the Shapley value, and other similar concepts like the
egalitarian Shapley value, even if we restrict to the subgame $v_{|S}$ where 
$S$ is the final coalition. This is because the Shapley value would consider
any ordering (scenario to reach $S$) as possible and equiprobable. Our
analysis shows that some scenarios may reveal infeasible, and scenarios have
probability of occurrence, reflecting externalities, that is, all the
possible transitions that can occur along the coalition formation process.

\subsubsection{With a cost of leaving ($c>0$)}

\label{sec:c>0} The previous section has shown that, in accordance with the
theoretical results, any scenario ends with a coalition of size 2. We now
show that introducing a cost of leaving permits to escape this well-known
phenomenon and make the grand coalition stable.

Introducing a cost $c>0$ can delete only backward transitions (i.e., from a
coalition to a subcoalition) in Figure~\ref{fig: trans graph}. Considering a
backward transition $S\rightarrow S\setminus\{i\}$, the minimal cost to put
for deleting this transition is simply equal to $c_{\min}=\pi_i(e^{S%
\setminus\{i\}})-\pi_i(e^S)$. Computing this quantity for all backward links
of Figure~\ref{fig: trans graph} yields

\begin{center}
\begin{tabular}{ll}
Transition & $c_{\min}$ \\ \hline
$14\rightarrow 1$ & 0.281 \\ 
$24\rightarrow 2$ & 0.281 \\ 
$34\rightarrow 3$ & 0.281 \\ 
$123\rightarrow 12$ & 0.143 \\ 
$123\rightarrow 13$ & 0.463 \\ 
$123\rightarrow 23$ & 0.143 \\ 
$124\rightarrow 12$ & 0.889 \\ 
$134\rightarrow 13$ & 0.889 \\ 
$234\rightarrow 23$ & 0.889 \\ 
$1234\rightarrow 123$ & 1.424%
\end{tabular}
\end{center}

It follows that if the cost is at least equal to 1.424, no backward arrow
exists any more in the transition graph, i.e., no country has incentive to
leave any coalition. The graph obtained with $c\geqslant 1.424$ is
represented on Figure~\ref{fig:trgr2}. 
\begin{figure}[h!]
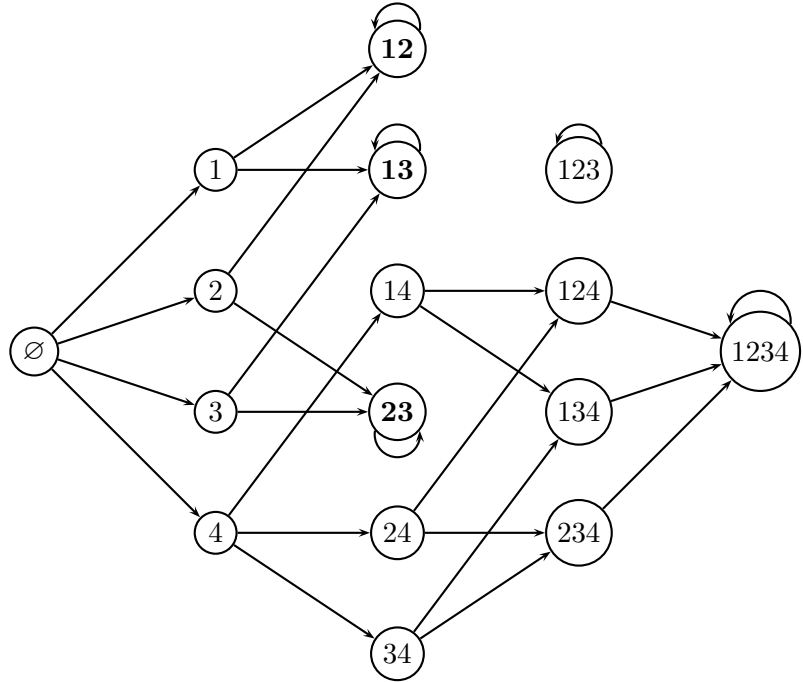

\begin{center}
\psset{unit=0.8cm} \pspicture(-0.2,-1)(12,11) \cnodeput(0,5){0}{$\varnothing$%
} \cnodeput(3,8){1}{$1$} \cnodeput(3,6){2}{$2$} \cnodeput(3,4){3}{$3$} %
\cnodeput(3,2){4}{$4$} \cnodeput(6,10){12}{$\mathbf{12}$} \cnodeput(6,8){13}{%
$\mathbf{13}$} \cnodeput(6,6){14}{$14$} \cnodeput(6,4){23}{$\mathbf{23}$} %
\cnodeput(6,2){24}{$24$} \cnodeput(6,0){34}{$34$} \cnodeput(9,8){123}{$123$} %
\cnodeput(9,6){124}{$124$} \cnodeput(9,4){134}{$134$} \cnodeput(9,2){234}{$%
234$} \cnodeput(12,5){1234}{$1234$} \ncarc[arcangle=0]{->}{0}{1} 
\ncarc[arcangle=0]{->}{0}{2} \ncarc[arcangle=0]{->}{0}{3} %
\ncarc[arcangle=0]{->}{0}{4} \ncarc[arcangle=0]{->}{1}{12} %
\ncarc[arcangle=0]{->}{1}{13} \ncarc[arcangle=0]{->}{2}{12} %
\ncarc[arcangle=0]{->}{2}{23} \ncarc[arcangle=0]{->}{3}{13} %
\ncarc[arcangle=0]{->}{3}{23} \ncarc[arcangle=0]{->}{4}{14} %
\ncarc[arcangle=0]{->}{4}{24} \ncarc[arcangle=0]{->}{4}{34} %
\ncarc[arcangle=0]{->}{14}{124} \ncarc[arcangle=0]{->}{14}{134} %
\ncarc[arcangle=0]{->}{24}{124} \ncarc[arcangle=0]{->}{24}{234} %
\ncarc[arcangle=0]{->}{34}{134} \ncarc[arcangle=0]{->}{34}{234} %
\ncarc[arcangle=0]{->}{124}{1234} \ncarc[arcangle=0]{->}{134}{1234} %
\ncarc[arcangle=0]{->}{234}{1234} \nccircle{->}{12}{0.3cm} %
\nccircle{->}{13}{0.3cm} \nccircle{->}{23}{-0.3cm} \nccircle{->}{123}{0.3cm} %
\nccircle{->}{1234}{0.4cm}\endpspicture
\end{center}
\caption{Transition graph with $c\geqslant 1.424$}
\label{fig:trgr2}
\end{figure}
Observe that now 1234 and 123 are new absorbing states, however the latter
is disconnected from the rest of the graph, hence not reachable. As backward
arrows have disappeared, transition probabilities have to be recomputed.
They are given in Table~\ref{tab: trans probc} below. 
\begin{table}[h!]
\caption{Transition probabilities with $c\geqslant 1.424$}
\label{tab: trans probc}
\begin{center}
\begin{tabular}{|c|c||c|c|}
\hline
$p_{\varnothing,1}$ & 0.488 & $p_{4,24}$ & 0.333 \\ 
$p_{\varnothing,2}$ & 0.256 & $p_{4,34}$ & 0.333 \\ 
$p_{\varnothing,3}$ & 0.066 & $p_{14,124}$ & 0.5 \\ 
$p_{\varnothing,4}$ & 0.190 & $p_{14,134}$ & 0.5 \\ 
$p_{1,12}$ & 0.5 & $p_{24,124}$ & 0.5 \\ 
$p_{1,13}$ & 0.5 & $p_{24,234}$ & 0.5 \\ 
$p_{2,12}$ & 0.5 & $p_{34,134}$ & 0.5 \\ 
$p_{2,23}$ & 0.5 & $p_{34,234}$ & 0.5 \\ 
$p_{3,13}$ & 0.5 & $p_{124,1234}$ & 1 \\ 
$p_{3,23}$ & 0.5 & $p_{134,1234}$ & 1 \\ 
$p_{4,14}$ & 0.333 & $p_{234,1234}$ & 1 \\ \hline
\end{tabular}%
\end{center}
\end{table}
Scenarios $\mathcal{S}_1$ to $\mathcal{S}_6$ from the case $c=0$ are
unaffected and still exist with the same probability of occurrence. However,
all scenarios $\mathcal{S}_7$ to $\mathcal{S}_{36}$ are replaced by new ones
which are:

\begin{itemize}
\item $\mathcal{S}^{\prime }_7$: $\varnothing\rightarrow 4\rightarrow
14\rightarrow  124\rightarrow 1234$

\item $\mathcal{S}^{\prime }_8$: $\varnothing\rightarrow 4\rightarrow
14\rightarrow  134\rightarrow 1234$

\item $\mathcal{S}^{\prime }_9$: $\varnothing\rightarrow 4\rightarrow
24\rightarrow  124\rightarrow 1234$

\item $\mathcal{S}^{\prime }_{10}$: $\varnothing\rightarrow 4\rightarrow
24\rightarrow  234\rightarrow 1234$

\item $\mathcal{S}^{\prime }_{11}$: $\varnothing\rightarrow 4\rightarrow
34\rightarrow  134\rightarrow 1234$

\item $\mathcal{S}^{\prime }_{12}$: $\varnothing\rightarrow 4\rightarrow
34\rightarrow  234\rightarrow 1234$
\end{itemize}

They all have the same probability, which is about 0.03167. Table~\ref%
{tab:ns} gives the scenario-values computed for the above new scenarios. 
\begin{table}[htb]
\caption{Scenario-values for $c\geqslant 1.424$}
\label{tab:ns}
\begin{center}
\begin{tabular}{|c|llll|}
\hline
scenario $\mathcal{S}$ & $\phi^\mathcal{S}_1$ & $\phi^\mathcal{S}_2$ & $\phi^%
\mathcal{S}_3$ & $\phi^\mathcal{S}_4$ \\ \hline
$\mathcal{S}^{\prime }_7$ & 11.813 & 7.131 & 4.369 & 4.500 \\ 
$\mathcal{S}^{\prime }_8$ & 11.813 & 8.869 & 2.631 & 4.500 \\ 
$\mathcal{S}^{\prime }_9$ & 12.631 & 6.313 & 4.369 & 4.500 \\ 
$\mathcal{S}^{\prime }_{10}$ & 14.369 & 6.313 & 2.631 & 4.500 \\ 
$\mathcal{S}^{\prime }_{11}$ & 12.631 & 8.869 & 1.813 & 4.500 \\ 
$\mathcal{S}^{\prime }_{12}$ & 14.369 & 7.131 & 1.813 & 4.500 \\ \hline
\end{tabular}
\end{center}
\end{table}

\paragraph{Interpretation of the results and discussion}

As we have seen above, there is no change for scenarios  $\mathcal{S}%
_1,\ldots,\mathcal{S}_6$, therefore the interpretation and discussion done
in  Section~\ref{sec:c=0} about these scenarios are still valid. A first 
observation concerning the new scenarios $\mathcal{S}^{\prime }_7,\ldots,%
\mathcal{S}^{\prime }_{12}$ is that  they are all feasible, for any country,
including country 4. As a consequence,  the expected value of the scenario
values over $\mathcal{S}^{\prime }_7,\ldots,\mathcal{S}^{\prime }_{12}$ is
an  imputation for all countries, i.e., they will get a payoff at least
equal to  their individual payoff:  
\begin{equation*}
\phi_1=12.938, \ \phi_2 = 7.438, \ \phi_3 = 2.938, \ \phi_4 = 4.500.  
\end{equation*}
Note also that these payoffs are at least equal to (and strictly greater
than for countries 1, 2, 3) the payoff they would obtain if the other
scenarios $\mathcal{S}_1,\ldots,\mathcal{S}_6$ realize. This means that
farsighted players should engage in the new scenarios $\mathcal{S}^{\prime
}_7,\ldots,\mathcal{S}^{\prime }_{12}$ rather than the other ones. These
scenarios lead to the formation of the grand coalition, and not to a
coalition of size 2.

We see that introducing a cost of leaving, which may be interpreted as a
damage to reputation, makes all countries better off (and strictly better
for countries 1, 2, 3). This cost must be at least equal to the payoff that
country 4 would benefit from leaving the grand coalition.

\section{Concluding remarks}

\label{sec:conclusions}

In this paper, we apply a Markov process approach to the formation of IEAs,
where any absorbing state of the Markov chain is a stable IEA. Our study
contributes to several streams of research focusing on IEAs: coalition
formation, and both cooperative and noncooperative game approaches. 
In particular, by introducing a cost of leaving a coalition, we show that 
the final coalition may have more than two members and may be the grand
coalition.  
Despite the vast related literature, the study of IEAs remains of high
importance.

There exists great potential for applying network theory to IEAs. \cite%
{Currarini et al.} discuss a wide range of application of network economics
to environmental problems, since they are typically related to local
interactions and network structures. However, contributions applying network
theory to IEAs seem still limited. \cite{GuntherHellmann17} study the
stability of IEAs in a repeated game framework, and characterize necessary
and sufficient conditions for stability of an IEA when pollution has global
and local effects. Local pollution spillovers are represented by a network,
where a link between two countries (e.g., sharing the same border),
indicates whether pollution of the countries affects each other. Clearly,
incorporating the network structure enriches considerably the analysis of
IEAs. In particular, \cite{GuntherHellmann17} show that, while stable IEAs
exist if the network structure is balanced, they might fail to exist for
networks being highly asymmetric with respect to neighbors in the network.

Applying network theory to IEAs is a natural follow-up of the present paper.
In particular, it would be interesting to apply a dynamic network formation
process introduced in \cite{CauGraRus15} to the analysis of IEAs. \cite%
{CauGraRus15} used a similar approach to that of \cite{fagr12,fagr13}, but
applied to network instead of coalition processes. The authors investigated
the dynamic random network formation processes, where links may appear and
disappear at any time, i.e., allowing that the network and the set of active
players may change over time. A scenario of the process, i.e., a sequence of
networks, is the result of a stochastic process, typically a Markov chain.
Following this dynamic network based approach, we could apply the link-based
scenario allocation rule introduced and axiomatically characterized in \cite%
{CauGraRus15}, to allocate value generated by the IEA dynamic network
formation process, i.e., to share benefits among all participants involved
in the IEA.

\section*{Statements and declarations}

\paragraph{Conflict of interest}
The authors have no competing interests to declare that are relevant to the
content of this article, nor any relevant financial or non-financial interests
to disclose. 

\paragraph{Funding}
Michel Grabisch and Agnieszka Rusinowska acknowledge the financial support
by the European Union HORIZON-WIDERA-2023-TALENTS-01 project No 101183743
(AGATE). Georges Zaccour acknowledges the financial support by NSERC,
Canada, grant RGPIN-2021-02462.


\end{document}